\documentclass[reqno,centertags,11pt,a4paper]{amsart}
\usepackage[left=2cm, right=2cm, top=2cm]{geometry}
\usepackage[utf8]{inputenc}
\usepackage[pdftex]{graphicx}
\usepackage{latexsym}
\usepackage{amsmath,amssymb,amsthm}
\usepackage{xcolor}
\usepackage{wrapfig}
\usepackage{hyperref}
\hypersetup{
    colorlinks,
    citecolor=blue,
    filecolor=black,
    linkcolor=blue,
    urlcolor=black
}

\newtheorem{Theorem}{Theorem}[section]
\newtheorem{Definition}[Theorem]{Definition}
\newtheorem{Lemma}[Theorem]{Lemma}
\newtheorem{Proposition}[Theorem]{Proposition}

\newtheorem{Bemerkung}[Theorem]{Remark}
\newtheorem{remark}[Theorem]{Remark}
\newtheorem*{assum*}{Assumptions}
\numberwithin{equation}{section}
\newcommand{\x}{x}
\newcommand{\xp}{y}
\newcommand{\C}{\mathbb{C}}

\newcommand{\R}{\mathbb{R}}

\newcommand{\N}{\mathbb{N}}

\newcommand{\be}{\begin{equation}}
\newcommand{\ee}{\end{equation}}
\newcommand{\halb}{\frac{\beta}{2}}

\newcommand{\la}{\begin{picture}(4.5,7)
\put(1.1,2.5){\rotatebox{60}{\line(1,0){5.5}}}
\put(1.1,2.5){\rotatebox{300}{\line(1,0){5.5}}}
\end{picture}}
\newcommand{\ra}{\begin{picture}(4.5,7)
\put(.9,2.5){\rotatebox{120}{\line(1,0){5.5}}}
\put(.9,2.5){\rotatebox{240}{\line(1,0){5.5}}}
\end{picture}}

\begin{document} %_______________________________________________________________
\title[]{On the van der Waals interaction between a molecule and a half--infinite plate}	

\author[]{Ioannis Anapolitanos,  Mariam Badalyan,	Dirk Hundertmark }
% Janni's address
\address{Department of Mathematics, Institute for Analysis, Karlsruhe Institute of Technology, 76131 Karlsruhe, Germany.}%
\email{ioannis.anapolitanos@kit.edu}%
% Mariam's address
\address{Department of Mathematics, Institute for Analysis, Karlsruhe Institute of Technology, 76131 Karlsruhe, Germany.}%
\email{uqeac@student.kit.edu}%
% Dirk's address
\address{Department of Mathematics, Institute for Analysis, Karlsruhe Institute of Technology, 76131 Karlsruhe, Germany, and Department of Mathematics,  
University of Illinois at Urbana-Champaign, 1409 W.\ Green Street (MC-382) 
Urbana, Illinois 61801, USA}%
\email{dirk.hundertmark@kit.edu}%

\date{\today, version \jobname }

\begin{abstract}  We consider a molecule in the Born-Oppenheimer approximation interacting with a plate of infinite thickness, i.e, a half--space, 
which is perfectly conducting or dielectric. It is well--known in the physics literature that in this case the atom or molecule is attracted by the 
plate at sufficiently large distances.  This effect is analogous to the well--known van 
der Waals interaction between neutral atoms or molecules. 
We prove that the interaction energy $W$  of the system is  
given by $W(r,v)= -C(v)r^{-3} + \mathcal{O}(r^{-4})$,
 where $r$ is the distance between the molecule and the plate and $v$ 
indicates their relative orientation. Moreover, $C(v)$ is positive 
and continuous, thus the atom or molecule is always pulled towards the 
plate at sufficiently large distances, for all relative orientations $v$. For some specific systems we provide sharper estimates of $W(r,v)$. 
This asymptotic behavior is well--known in the physics literature,  
however, we are not aware of any previous rigorous results, even on existence of a ground state of the system. 
  For pedagogical reasons, we often start with the case of a hydrogen atom and then we generalize the arguments to deal with a general molecule. 
 \end{abstract}
\thanks{\copyright 2020 by the authors. Faithful reproduction of this article, in its entirety, by any means is permitted for non-commercial purposes}

\maketitle
{\hypersetup{linkcolor=black}
\setcounter{tocdepth}{1}
\tableofcontents}

%__________________________________________________________________________________________________________________________________________________________________________________________
  % Ab sofort Seitenzahlen in der Kopfzeile anzeigen
 % \pagestyle{headings}

\section{Introduction}

Van der Waals  forces are usually studied between atoms or molecules. 
They are weaker than ionic or covalent bonds and decay rapidly with distance. 
Due to their universal nature, they play an important role in many different 
fields such as physics, quantum chemistry or material sciences and are 
important for the macroscopic properties and physical behavior of 
numerous materials. For example, they can significantly  
influence melting and boiling temperatures.  
They explain why diamond, which 
consists  of carbon atoms that are connected only with covalent bonds, is 
a much harder material than graphite, which consists of layers of carbon 
atoms that attract each other through van der Waals forces, see \cite{Cha}. 
For this reason, they have been studied extensively in the physics 
literature, see e.g. some classical works  \cite{Feinberg},  \cite{Jones1}, \cite{Jones2}, \cite{Lo},   \cite{vdW1}, \cite{vdW2} and some more recent \cite{Beguin}, \cite{Bjork}, \cite{DiStatio}, \cite{sutter}, to name a few.

%\begin{wrapfigure}{r}{0.3\textwidth}
%  \centering
% \includegraphics[width=0.3\textwidth]{Giant}
%\tiny Bildquelle: Wikipedia
%\end{wrapfigure}

%\begin{center}
%\includegraphics[scale=0.5]{seta1.jpg}\\
%\tiny Bildquelle: \cite{autumn1}
%\end{center}

%It is thus unsurprising, that a number of researches are %examining geckos, in order to develop new and stronger %adhesives. 
%Professor Stanislav Gorb and his team from the %Christian-Albrechts-Universität zu Kiel for example, have %created such an adhesive. One 20cm x 20cm piece of this %reusable tape, once fixed to glass, can support the weight of %a grown man \cite{papergeim}.  However, the research continues %and the number of possible applications is endless from %industrial adhesives to new and better tissue adhesives %\cite{bandaidpaper}. \\

As a consequence, it is important to look at the van der 
Waals forces also from a theoretical and even mathematically rigorous point of view. 
%The rigorous study of the  "van der Waals problem" started 
%about forty years ago. 
J.~D.~Morgan and B.\ Simon proved in 1980 the 
existence of an asymptotic expansion of the van der Waals interaction energy 
of two neutral atoms in powers of one over the distance of their nuclei 
for large nuclear separations \cite{morgansimon}, assuming non-degeneracy 
of the ground state energies of the atoms. However, they did not identify the 
coefficients in such an expansion, in particular, it is not clear from their 
method that the leading order coefficient responsible for the van der Waals 
type attraction is non-zero. 
Later,  E.\ H.\ Lieb and W.\ E.\ Thirring used variational methods 
to prove an \emph{upper bound} for the interaction energy, which shows  
the existence of van der Waals forces for Coulomb systems 
\cite{liebthirring}, like systems of molecules. 
Roughly two and a half decades later,  I.M.\ Sigal and the first author 
\cite{anapolitanossigal} provided rigorously the leading term of  the 
long range behavior of the van der Waals interaction energy between atoms  under some 
conditions.  In \cite{paperjannis} the methods of \cite{anapolitanossigal} were refined in order to study the long range asymptotics coupled with the asymptotics of the number of atoms going to infinity.  M. Lewin, M. Roth and the first author investigated in \cite{anapolitanoslewinroth} the derivative of the interaction energy of two atoms, which in physics is interpreted as force, and provided the leading order of the interaction energy with no assumptions on the  ground state eigenspace for one of the two atoms. 
M.\ Lewin and the first author improved in \cite{anapolitanoslewin} the 
upper bound of Lieb and Thirring,  and under some assumptions provided 
the leading terms of the long range behavior of the van der Waals forces 
between molecules and they used these results to study isomerizations.  
 Recently, the van der Waals forces between atoms were investigated in 
 the case of semi-relativistic kinetic energy in \cite{barbaroux} by 
 J.M. Barbaroux, M. Hartig, S. Vugalter and the third author. They also 
 rigorously proved the famous Axilrod--Teller--Muto $D^{-9}$ three body 
 correction, see \cite{axilrod-teller} and \cite{muto}, which plays an important role in the case of three or 
 more interacting atoms.
 Note that in the case of two hydrogen atoms the leading term coefficient 
 of the interaction energy was approximated numerically by E. Cancès and 
 L.R. Scott in \cite{paper2hyd} with a proof of convergence of the 
 numerical scheme. Their numerical scheme is based on a modification of 
 a technique introduced in \cite{KSl}. The proof of convergence of their 
 numerical scheme  partially uses methods of \cite{anapolitanossigal}. More recently E. Cancès, R. Coyaud 
and L.R. Scott studied coefficients of higher order terms of the interaction energy of two hydrogen atoms numerically in  \cite{paper2hydhigherorders}.

The specific problem of the van der Waals interaction between a particle and a dielectric or perfectly conducting  plate is of interest to physicists as well and has been studied in the physics literature, see e.g. \cite{bardeen}, \cite{mavroyannis}, \cite{Bezerra}, \cite{Caride}, \cite{papercasimirpolder}, \cite{Mahanty}, \cite{peyrot}. One example of this type of interaction is the deflection of beams of particles by uncharged surfaces, see \cite{raskinkusch}, \cite{shih}. Another one is  the adhesion power of a gecko's foot \cite{autumn2}. These reptiles are famous for their ability to adhere to very smooth surfaces without  adhesives like glue or suction cups but thanks to the very special structure of their feet. Several studies e.g. \cite{autumn1}, \cite{autumn2}, have claimed that the main contribution to this is due to the van der Waals force.
A more recent study  claims that the main contribution to this force comes from  contact electrification \cite{izadi}, nevertheless the van der Waals force still contributes to it. 

In this article, we look at a system consisting of a molecule in the vacuum interacting with a perfectly conducting or dielectric plate, more precisely, with a half--space, and  estimate its interaction energy $W(r,v)$, depending on the distance $r$ of the molecule to the plate and on their relative orientation $v$. This regime creates technical difficulties as it is modelled with Schrödinger operators for systems of atoms acting on a half--space. Thus, fundamental properties like existence of a ground state of the system and identification of its essential spectrum were, to the best of our knowledge, unknown  and part of this work is to study them. On the other hand this regime is at the same time easier because mirror charges are strictly correlated to the true ones in contrast e.g. to the case of two interacting atoms or molecules. This simplicity allows us to obtain stronger results, most importantly we provide the leading term of the interaction energy with no assumptions on the ground state eigenspace of an atom/molecule, something that hasn't been done in any previous work in the field.
\smallskip 

 The paper is organized as follows. In Section \ref{modelinghydrogenwall} we introduce and derive the 
Hamiltonian describing the hydrogen--dielectric  plate system.  The case of a perfectly conducting plate is derived in an appropriate limit.  We then derive the Hamiltonian for a general molecule  interacting with the dielectric plate and explain some of its basic properties.
 Afterwords we state a binding Condition, see equation \eqref{molvermutung} below, which we assume and is in fact a necessary condition for our result. We also explain why \eqref{molvermutung} is physically expected. We state then our main Theorem, which in the special case of a hydrogen atom is sharper. As we explain below, conditions that are somewhat analogous to \eqref{molvermutung} had to be assumed to derive the van der Waals law for a system of atoms. 
 In Section \ref{preparations}
we prove basic properties of the system, e.g. that the Hamiltonian of the system is bounded from below and for the case of a hydrogen atom interacting with the plate we prove a Weyl type theorem for its essential spectrum. This helps us prove that the ground state energy is below the essential spectrum, in particular, the Hamiltonian has a ground state, at least when the nucleus is not too close to the plate.  We moreover prove \eqref{molvermutung} for the special cases of a helium and a hydrogen atom. 
 Furthermore, we introduce the Feshbach map which is a main ingredient of the proof. In Section \ref{Beweis} we prove the main Theorem for the case of the hydrogen atom. In Section \ref{generalizationmolecule} we generalize the proof done for a hydrogen atom to the case of a general molecule. Since many ideas are similar in this case, we focus on the modifications of the proof. In Section \ref{proofnecessity}
 we sketch the proof of the fact that \eqref{molvermutung} is a necessary Condition for the main Theorem. 
  Finally, in Appendix \ref{spiegelladung} we provide some detailed explanations concerning the derivation of the Hamiltonian of the system. 
\newline

\noindent {\bf Acknowledgments.} It is a pleasure to thank 
	Kurt Busch and Francesco Intravaia for suggesting the problem, Semjon Vugalter for inspiring discussions,  Carsten Rockstuhl for helpful discussions on the physics of the problem, and  Francois Cornu for discussing his work \cite{cornumartin}. We gratefully acknowledge an anonymous referee for numerous useful comments that helped us improve our work. 
 We also thank the Deutsche Forschungsgemeinschaft (DFG, German Research Foundation) for financial support -- Project-ID 258734477 -- SFB 1173.

%----------------------------------------------------------------------------------------------------------------------------
%-------------------------------- MODELLIERUNG DES PROBLEMS ----------------------------------------------------
%----------------------------------------------------------------------------------------------------------------------------

\section{Modeling of the problem and main result}\label{modelinghydrogenwall}

We first model the problem starting with the special case of a hydrogen atom and afterwards for a general molecule.

\subsection{Hamiltonian in the special case of a hydrogen atom}

We will follow \cite{cornumartin} to model the problem.
We first consider a hydrogen atom and a perfectly conducting plate placed in a vacuum. Let $(0,0,0)^t$ be the position of the nucleus and $x=(x_1,x_2,x_3)^t\in\R^3$ be the position of the electron. Further we assume without loss of generality that the plate is orthogonal to $\vec{e}_1=(1,0,0)$, and passes through $r\vec{e}_1$, hence $r$ is the distance between the nucleus of the hydrogen atom and the plate. Moreover, the plate is on the right of the atom, see Figure 1. \\

\begin{center}
\includegraphics[scale=0.3]{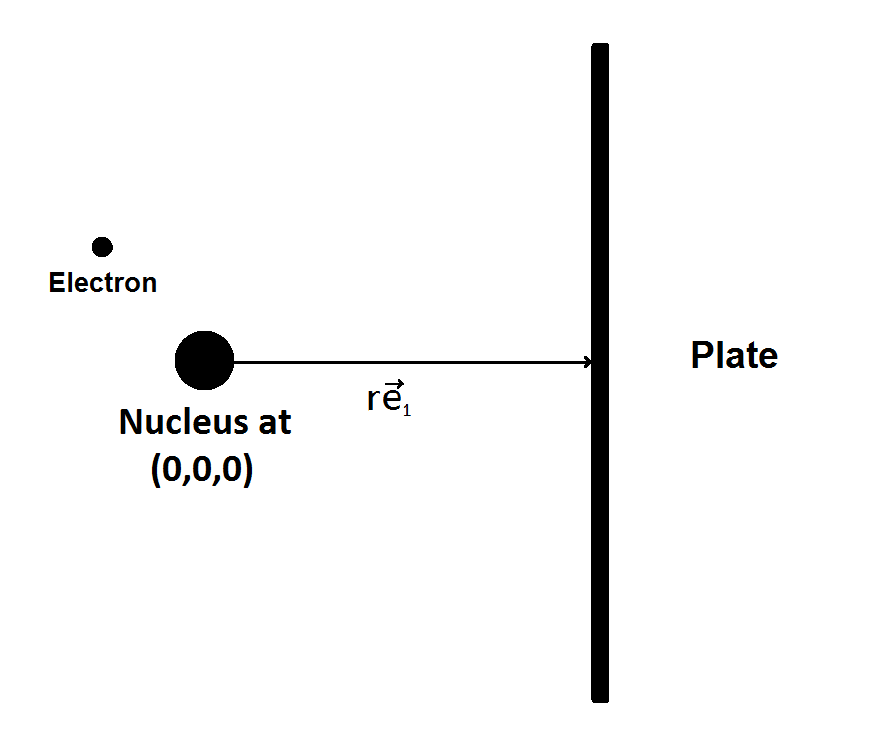}\\
\small Figure 1 
\end{center}

Because the van der Waals forces decay rapidly with the distance, the atom has to be very close to the plate for these forces not to be negligible. Furthermore, the height, width and thickness of the plate are so large in comparison to the distance of the atom to the plate, that we can assume they are infinite, unless the plate is extremely thin, a case that we do not investigate in the present work.

As usual, we work with the Born-Oppenheimer approximation, which assumes that the nucleus is at a fixed position. This approximation relies on the fact that the mass of the nucleus is much larger than the mass of the electron. %\footnote{\label{massenprotelect}$m_{proton}\approx 1,67\times 10^{-27}kg$, $ m_{electron}\approx 9,1\times 10^{-31}kg$ }.
 For a discussion of the validity of the Born-Oppenheimer approximation see e.g. \cite{anapolitanossigal} and references therein.

 The Hamiltonian as a function of the distance $r$  is given by the sum of the kinetic energy of the electron and the interaction energies (potential energy) of the system:

\be \begin{split}
	\widehat{H}&=\widehat{T}+\widehat{V}\\
	&= -\frac{\hbar^2}{2m}\Delta_{x}+V(r)
\end{split}\ee

In order to determine the potential energy $V(r)$ of our system, we use the method of image charges, which guarantees that the potential on the surface of the plate is zero. If the plate is a perfectly conducting  plate, then  we just have to introduce, for the electron at $x$,  a "positive charge" ($e^+$) at the mirror image, with respect to the surface of the half--space, at $2re_1+x^*$, where $x^*=(-x_1, x_2, x_3)^t$  and for the nucleus a "negative charge" ($K^-$) in $(2r,0,0)^t$, see Figure 2.

\begin{center}
\includegraphics[scale=0.3]{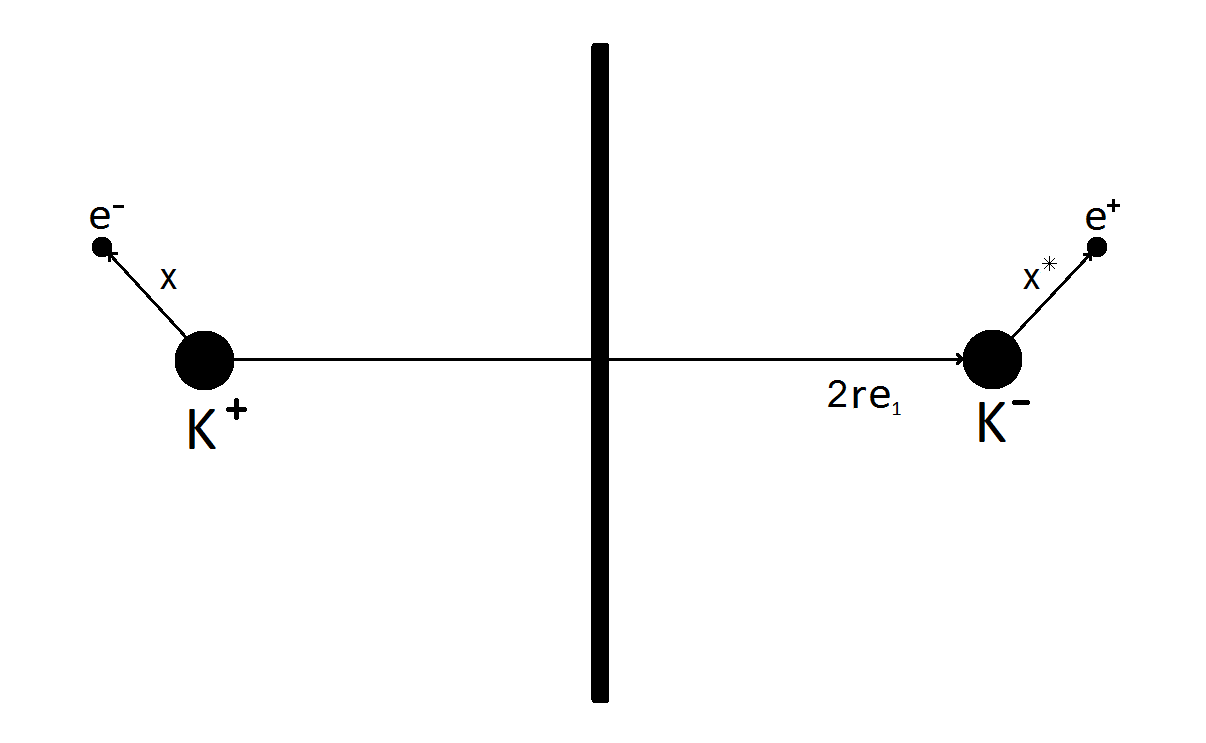}\\
\small Figure 2 
\end{center}

 Thus we  obtain a system of point charges and the sum of their Coulomb potentials is the potential energy $V(r)$. 
 The Hamiltonian of the system in the non-relativistic case is then given by 

\begin{align*}
H(r)=&-\frac{\hbar^2}{2m}\Delta_{x}-\underbrace{\frac{ e^2}{4\pi\epsilon_0\vert x\vert}}_{\substack{attraction \\e^{-}/ K^+}}+\halb\bigg(\underbrace{-\frac{ e^2}{4\pi\epsilon_0(2r)}}_{\substack{attraction\\ K^+/K^-}}-\underbrace{\frac{ e^2}{4\pi\epsilon_0\vert -x +2re_1 + x^*\vert}}_{\substack{attraction\\e^{-}/e^{+}}}\\
&+\underbrace{\frac{ e^2}{4\pi\epsilon_0\vert 2re_1-x\vert}}_{\substack{repulsion\\e^{-}/ K^{-}}}+\underbrace{\frac{ e^2}{4\pi\epsilon_0\vert 2re_1 + x^*\vert}}_{\substack{repulsion\\e^{+}/ K^+}}\bigg),
\end{align*}
where $\hbar$ is the reduced Planck constant, $m$ the mass of the electron, $e$ the elementary charge and $\epsilon_0$ the vacuum permittivity. For a general dielectric plate we have
\begin{equation}\label{dielectriccon}
 \beta:=\frac{\epsilon-\epsilon_0}{\epsilon+\epsilon_0}
\end{equation}
  where $\epsilon>\epsilon_0$ is the permittivity of the plate. The perfectly conducting case corresponds to the limit $\epsilon \to \infty$ and in this case $\beta=1$. In general we have 
\begin{equation}\label{betarange}
\beta \in (0,1].
\end{equation}
 We refer to Appendix \ref{spiegelladung} for more explanations. 

We scale a normalized wave function in $ \psi  \in L^2(\R^3)$ by $\psi_\alpha=\alpha^{\frac{3}{2}}\psi(\alpha x)$, so 
$\psi_\alpha$ is  normalized as well. Then
\be\label{res1}
\la\psi_\alpha|\Big(-\Delta-\frac{\alpha}{|x|}\Big)\psi_\alpha\ra=\alpha^2\la\psi|\Big(-\Delta-\frac{1}{|x|}\Big)\psi\ra,
\ee
and consequently we obtain
\be\label{res2}
\inf\sigma\Big(-\Delta-\frac{\alpha}{|x|}\Big)=\alpha^2\inf\sigma\Big(-\Delta-\frac{1}{|x|}\Big).
\ee
Based on this type of argument we can simplify our computations, setting  $\frac{\hbar^2}{2m}=1$ and $\frac{ e^2}{4\pi\epsilon_0}=1$. In that case  we find 

\be \nonumber
H(r)=-\Delta_{x}-\underbrace{\frac{1}{\vert x\vert}}_{\substack{attraction \\e^{-}/ K^+}} + \halb\bigg(-\underbrace{\frac{1}{2r}}_{\substack{attraction\\ K^+/K^-}}-\underbrace{\frac{1}{\vert -x+2re_1+x^*\vert}}_{\substack{attraction\\e^{-}/e^{+}}} $$$$+\underbrace{\frac{1}{\vert 2re_1 -x\vert}}_{\substack{repulsion\\e^{-}/ K^{-}}}+\underbrace{\frac{1}{\vert 2re_1 +x^* \vert}}_{\substack{repulsion\\e^{+}/ K^+}}\bigg).
\ee

Because of the symmetry of the problem $\vert 2re_1 + x^* \vert=\vert 2re_1-x\vert$ and the Hamiltonian becomes 
\vspace*{0.5cm}
\be\label{defHlin} H(r)=-\Delta_{x}-\frac{1}{\vert x\vert}+ \halb U(x),\ee
where 
\be\label{defU} U(x)=-\frac{1}{2r}-\frac{1}{\vert 2re_1-(x-x^{*})\vert}+\frac{2}{\vert 2re_1-x\vert}.\ee

Observe that the electron-mirror electron attraction term $-\frac{1}{\vert 2re_1-(x-x^{*})\vert}$ is not locally $L^1$ coming from the fact that the function $\frac{1}{e_1 \cdot z}$ is not locally integrable in $\R_+ \times \R^2$. Thus without any boundary conditions $H(r)$ would not be bounded from below. 
As we will see below, the Hamiltonian $H(r)$ acts on wavefunctions on the halfspace satisfying Dirichlet boundary condition, which ensures that the ground state energy is not $-\infty$. 
 
The potential $U(x)$ is a dipole-dipole interaction and comes from orientations of attraction and it is clear that it is attractive for large distances. 
That the potential  $U(x)$ is always attractive  is seen by the following lemma:
\begin{Lemma}\label{lemmatrapezoid}
The potential $U(x)$ given in \eqref{defU} satisfies $U(x) \leq 0$ with equality if and only if $x=0$.
\end{Lemma}
\begin{proof}
		We consider an isosceles trapezoid with diagonal length $b$ and let $a,c$ be the lengths of the parallel sides. Then
		\begin{equation}\label{ineqtrap}
		\frac{2}{b} \leq \frac{1}{a} + \frac{1}{c}.
		\end{equation}
	Indeed, assume without loss of generality that $a \geq c$.
	Since $b$ is the diagonal we have that $b \geq \frac{a+c}{2}$. One can see that by bringing a height $h$ of the trapezoid to the side $c$ as in Figure 3 below. This creates an orthogonal triangle with hypotenuse $b$ and $\frac{a+c}{2}$ one of the other sides.  Since $a+c \geq 2 \sqrt{ac}$, we obtain that $b \geq \sqrt{ac}$ or that $\sqrt{\frac{b}{a}}  \sqrt{\frac{b}{c}} \geq 1$. This in turn gives that $\frac{b}{a} + \frac{b}{c} \geq 2$ as desired.
	\begin{center}
		\includegraphics[scale=0.3]{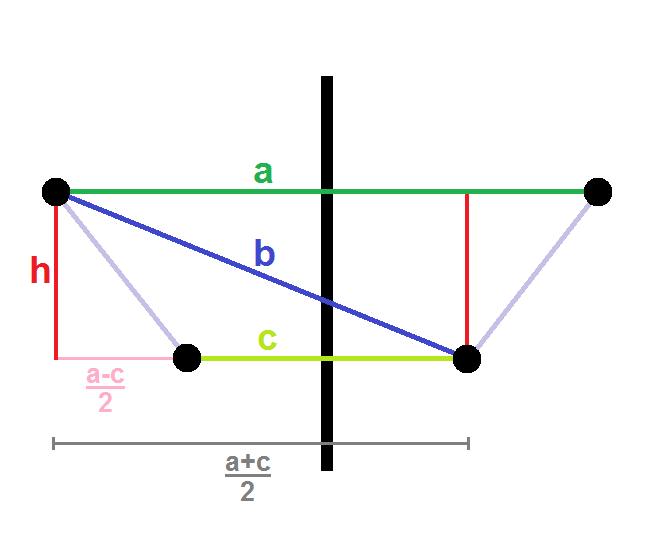}\\
		\small Figure 3: Isosceles trapezoid formed by the  electron  the nucleus and their mirror charges. 
	\end{center}
From \eqref{ineqtrap} and \eqref{defU} it follows immediately that $U(x) \leq 0$. To see when equality holds it is enough to see when equality in \eqref{ineqtrap} holds.
From the proof of \eqref{ineqtrap}  one easily sees that if equality holds in  \eqref{ineqtrap} then one must have $b=a+c$, 
so all vertices of the trapezoid have to be on the same line. 
Moreover we must have $a=c$. So actually the potential added in \eqref{defHlin} to the Hamiltonian of the hydrogen atom is strictly negative with exception only the case that the position of the electron coincides with the position of the nucleus, namely when $x=0$. This completes the proof of Lemma \ref{lemmatrapezoid}.
\end{proof}

\subsection{Hamiltonian for a general molecule}

 We consider a molecule with nuclei in the positions $y_1,...,y_M \in \R^3$, with atomic numbers $Z_1,...,Z_M$ and with $N$ electrons. Due to neutrality we impose 
 \begin{equation}\label{neutral}
 N=\sum_{j=1}^M Z_j.
 \end{equation}
 We work with the Born-Oppenheimer approximation and we assume without loss of generality
 \begin{equation}\label{centermass0}
 \sum_{j=1}^M Z_j y_j=0.
 \end{equation}
 After the same change of units as in the case of one hydrogen atom we find that
 the Hamiltonian of the molecule without the plate is given  by 
 \begin{equation}\label{HNdef}
 H_N=\sum_{i=1}^N \left( -\Delta_{x_i} - \sum_{k=1}^M \frac{Z_k}{|x_i-y_k|}\right) + \sum_{1 \leq i < j \leq N} \frac{1}{|x_i-x_j|}.
 \end{equation}  
 Note that the sum of the $\sum_{1 \leq k< l \leq M} \frac{Z_k Z_l}{|y_k-y_l|}$ of the repulsion terms between the nuclei, which is normally present, is omitted here. The reason is that if we omit it, the interaction energy of the system, defined below in \eqref{interen}, does not change.  For the proof of our results this omission turns out to be convenient. 
 If a dielectric plate is placed vertically to a unit vector $v$ and passes through $rv$ then the Hamiltonian of the full system can be derived in a similar way as in the case of a hydrogen atom, involving now the interaction terms between electrons-mirror electrons, electrons-mirror nuclei, nuclei-mirror electrons, nuclei-mirror nuclei, see 
 Appendix \ref{spiegelladung}. 
 Here $v$ cannot be chosen to be $e_1$ because the molecule is not rotationally symmetric. Of course we have to impose that all nuclei are on the same side of the plate which mathematically means that
 \begin{equation}\label{kernbed}
 y_k \cdot v <r, \qquad \forall k \in \{1,\dots,M\}.
 \end{equation}

 The Hamiltonian of the system can be derived with similar arguments as in the derivation of \eqref{defHlin} and is given by 
 \begin{equation}\label{Hdec}
 H=H(r,v)=H_N+\halb I,
 \end{equation}
 where $H_N$ is given by \eqref{HNdef} and
 \begin{equation}\label{Idec}
 I=I_1 - I_2 - I_3,
 \end{equation}
 with
 \begin{align}\label{eqI1}
 I_1=\sum_{i=1}^N \sum_{l=1}^M   \frac{2 Z_l}{|-x_i+2r v +y_l^*|},\\ \label{eqI2} I_2=\sum_{i=1}^N \sum_{j=1}^N \frac{1}{|-x_i + 2r v + x_j^*|}, \\ \label{eqI3} 
  I_3= \sum_{k=1}^M \sum_{l=1}^M  \frac{Z_k Z_l}{|-y_k + 2rv + y_l^*|}.
 \end{align}
 Here $^*$ stands for reflection with respect to the plane which is orthogonal to the vector $v$ and passes through 0.
 The term $I_1$ consists of the repulsion terms of the electrons from the mirror nuclei, and the nuclei from the mirror electrons, which, due to symmetry, turn out to be in pairs the same and this is why we multiply with 2, similarly as in the case of a hydrogen atom. 
  $I_2$ consists of the attraction terms of the electrons from the mirror electrons and $I_3$ of the attraction terms of the nuclei from the mirror nuclei.   We conjecture that we always have $I \leq 0$ and $I<0$ almost everywhere, but we could prove this only in the special case of a hydrogen atom, see Lemma \ref{lemmatrapezoid} above.

 \subsection{Hilbert space, boundedness from below and self-adjointness of the Hamiltonian}

For a general  molecule each electron is located on the half space
\begin{equation}\label{def:R3rv}
\R^3_{r,v}=\{x \in \R^3: x \cdot v \leq r\},
\end{equation}
where the restriction imposes, similarly to \eqref{kernbed}, that the electrons are all on the same side of the plate. 
  Due to the fermionic nature of the electrons, the Hamiltonian $H=H(r,v)$ defined in \eqref{Hdec} acts in the Hilbert space
  \begin{equation*}
  \mathcal{H}=L^2_a\left((\R^3_{r,v}\times\{\pm1/2\})^N,\C\right)\simeq \bigwedge_1^NL^2\left(\R^3_{r,v}\times\{\pm1/2\},\C\right)
  \simeq \bigwedge_1^NL^2\left(\R^3_{r,v},\C^2\right)
  \end{equation*}
  of antisymmetric square-integrable wave functions 
  $\Psi(x_1,s_1,\dots ,x_N,s_N)$ with spin, that is, such that
  \begin{equation}\label{eq:fermions}
  \Psi(X_{\pi(1)},\dots ,X_{\pi(N)})=(-1)^\pi\,\Psi(X_{1},\dots ,X_{N}) 
  \end{equation}
  for any permutation $\pi\in S_N$, where $X=(x,s)\in \R^3\times\{\pm1/2\}$. As  we previously explained, when we discussed the special case of the hydrogen atom, $H(r,v)$ is not bounded from below if we do not impose any boundary conditions. 
  To achieve that $H(r,v)$ is bounded from below we impose Dirichlet boundary  condition  (see for example \cite{cornumartin}), so we choose as form domain of $H(r,v)$ the Sobolev space 
  \begin{equation}\label{formdomain}
   \mathcal{H}_{0,N}= L^2_a\left((\R^3_{r,v}\times\{\pm1/2\})^N,\C\right) \cap H^1_0\left((\R^3_{r,v}\times\{\pm1/2\})^N,\C\right). 
  \end{equation}
  The meaning of the Dirichlet boundary condition in the choice of the space is that the electrons can not pass through the plate or touch the plate.  The ground state energy $E(r,v)$ of the system is then defined by 
  \be 
  \label {defE}
  E(r,v)=\inf \sigma(H(r,v))
  =\inf_{\psi\in \mathcal{H}_{0,N}, \| \psi \|_{ L^2}=1} \la\psi\mid H(r,v)\psi\ra.
  \ee 
  We prove in Section \ref{infimumnichtunendlich} that $E(r,v)$ is  well-defined and bigger than $-\infty$ and that $H(r,v)$ can be realized as a self-adjoint operator with form domain $\mathcal{H}_{0,N}$ and operator domain $\mathcal{H}_{0,N} \cap H^2\left((\R^3_{r,v}\times\{\pm1/2\})^N,\C\right)$.
\smallskip

We note that our result turns out not to depend on the statistics of the particles, nor on the presence of the spin, but we consider this case for obvious physical reasons.

\subsection{A necessary binding Condition}

  When studying van der Waals forces between atoms or molecules, one has to assume that the
electrons prefer to distribute themselves in a locally neutral way,  see e.g. Equation (E) in the introduction of \cite{anapolitanossigal}. In other words when the nuclei are far from each other, the electrons prefer to distribute themselves so that we have a system of (neutral) atoms or molecules and not a system of ions. This assumption is physically obvious but from a mathematical point of view a famous open problem. It turns out to be a necessary condition for the van der Waals law, since oppositely charged ions attract with Coulomb interaction. 

For our results we need a binding Condition which is roughly saying that when the nuclei are far from the plate the electrons prefer to stay close to the nuclei rather than moving close to the plate.
 The precise formulation is given below in \eqref{molvermutung}. As we explain below it is physically expected that \eqref{molvermutung} holds and it is a necessary Condition for our result. 
  To motivate \eqref{molvermutung} we will first explain some properties of a system of electrons with the plate.

Let 
\begin{equation}\label{Heminus}
H_{e^-}=-\Delta_{z}-\frac{1}{2|z-z^*|}= -\Delta_{z}-\frac{1}{4z_1}
\end{equation}
 be the Hamiltonian of the system consisting of the plate and just one electron, where $z:=x-re_1$, in the perfectly conducting case ($\beta=1$). The operator $H_{e^-}$ acts on $H_0^1(\R_+ \times \R^2)$. Then
\be\label{Eenegativ}
E_{e^-}:=\underset{\psi\in H_0^1(\R_+ \times \R^2), \| \psi \|_{ L^{2}(\R_+ \times \R^2)}=1}{\inf} \la\psi\mid H_{e^-}\psi\ra<0
\ee 
because for $\psi_\alpha(y)=\alpha^{\frac{3}{2}}\psi(\alpha y)$ with $\alpha>0$:
\be\begin{split}
	\label{Skalierung}
	\la\psi_\alpha|H_{e^-}\psi_\alpha\ra&=\la\psi_\alpha|\Big(-\Delta_z-\frac{1}{2|z-z^*|}\Big)\psi_\alpha\ra\\
	&=\alpha^2\la\psi|-\Delta_z\psi\ra-\alpha\underbrace{\int \frac{|\psi(z)|^2}{2|z-z^*|}dz}_{\substack{>0}}\\
	&<0 \hspace{2mm}\mbox{ for } \alpha \mbox{ small enough.}
\end{split}
\ee
We also have $E_{e^-}>-\infty$, see Section \ref{abschaetzungfuerEe}. Note that using \eqref{Skalierung} for $\alpha=\beta \in (0,1]$ we find that
\be
\label{Eescaling}
\underset{\psi\in H_0^1(\R_+ \times \R^2), \| \psi \|_{ L^{2}(\R_+ \times \R^2)}=1}{\inf} \la\psi\mid |\Big(-\Delta_z-\frac{\beta}{2|z-z^*|}\Big)\psi\ra = \beta^2 E_{e^-}. 
\ee

For all $k \in \{1,...,N\}$ we define the Hamiltonian of $k$ electrons with the plate by
\begin{align}\label{Ak}
A_k&=-\sum_{i=1}^k \Delta_{x_i}  + 	V_1(x),
\end{align}
where 
	\begin{align}\label{V1x}
	V_1(x)= \sum_{1\leq i < j \leq k} \frac{1}{|x_i-x_j|} -\halb \sum_{i=1}^k \sum_{j=1}^k \frac{1}{|-x_i + 2 r v  +x_j^*|},
	\end{align}
	is the potential of  $k$ electrons interacting with each other, with their own mirror charges, and with the mirror charges of the other electrons. As a form domain we consider the space $\mathcal{H}_{0,k}$ defined in \eqref{formdomain}.
 Let $Q_m$ denote the orthogonal projection onto the antisymmetric functions of $m$ particles with respect to exchanges of position-spin pairs.
More explicitly,
\be\label{def:Qm}
Q_m \Phi(X_1,\dots,X_m)= \frac{1}{m!} \sum_{\pi \in S_m} (-1)^{\pi} \Phi(X_{\pi(1)}, \dots, X_{\pi(m)}).
\ee   

 We now prove with the help of \eqref{Eescaling} the following
\begin{Lemma}\label{lem:k electrons plate}
	The spectrum of $A_k$ is given by $\sigma(A_k)= \sigma(A_k Q_k) = [k \beta^2 E_{e^-},\infty)$.
\end{Lemma}

\begin{proof}%[Proof of Lemma \ref{lem:k electrons plate}]
 Let 
	\begin{align*}
	V_2(x)=  -\halb \sum_{i=1}^k \frac{1}{|-x_i + 2 r v + x_i^*|},
	\end{align*}
	be the potential of  $k$ electrons interacting only with their  own mirror charges. 
	Then since $\beta \in (0,1]$, using \eqref{V1x} we find
	\begin{align*}
	V_1(x)-V_2(x) 
	\geq \halb \sum_{i\neq j } 
	\left(  
	\frac{1}{|x_i-x_j|}-  \frac{1}{|-x_i + 2 r v  +x_j^*|}
	\right)\, .
	\end{align*}
	For each pair of electrons $x_i, x_j$ are in the same half--space. As a consequence,
	the distance $|x_i-x_j|$ from $x_i$ to $ x_j$ is clearly 
	smaller than the distance $|-x_i + 2 r v  +x_j^*|$ from  
	$x_i$ to the mirror of $ x_j$. 
	Thus $ V_1(x)-V_2(x)\ge 0 $, which together with \eqref{Ak}  gives
	\begin{align*}
	A_k \ge \widetilde{A}_k:= -\sum_{i=1}^k \Delta_{x_i} +  V_2(x) 
	\end{align*}
	which is the Hamiltonian describing $k$ non--interacting electrons in the presence of  the plate. 
	Clearly by \eqref{Eescaling} $\sigma (\widetilde{A}_k)=\sigma (Q_k \widetilde{A}_k) = [k \beta^2 E_{e^-},\infty)$, hence  
	\begin{align*}
	\sigma(A_k)\subset \sigma (\widetilde{A}_k), \quad  \sigma(Q_k A_k)\subset \sigma (Q_k \widetilde{A}_k).
	\end{align*}
	On the other hand, placing $k$ electrons far away from each other, shows that we also have the reverse inclusion
	\begin{align*}
	\sigma(A_k)\supset \sigma (\widetilde{A}_k), \quad  \sigma(Q_k A_k)\supset \sigma (Q_k \widetilde{A}_k) 
	\end{align*}
	which proves Lemma \ref{lem:k electrons plate}. 
\end{proof}

A binding condition, which    guarantees the existence of a ground state of the molecule-plate system if the distance $r$ is not too small, is 
\begin{align} \label{molvermutung}
\inf \sigma(Q_N H_N)< & \inf\sigma(Q_{N-k} H_{N-k}) 
+ k \beta^2 E_{e^-} , \quad \forall k \in \{1,\dots, N\}.
\end{align}
In light of Lemma \ref{lem:k electrons plate}
the physical meaning of condition   \eqref{molvermutung} is that when the plate and the nuclei of the molecule are far from each other, it is energetically favorable for the electrons to be close to the nuclei rather than close to the plate. In Section \ref{abschaetzungfuerEe} we prove the validity of \eqref{molvermutung} for the cases of a hydrogen atom and of a  helium atom. 
 We expect \eqref{molvermutung} to hold in general.
 In fact in the special case that the molecule is an atom there is a clear experimental evidence that it holds. 
 Indeed,  the difference $\inf\sigma(Q_{N-1} H_{N-1})- \inf\sigma(Q_{N} H_{N})$
 is the first ionization energy of the molecule, namely the energy needed to remove an electron from the molecule.  
 It is well known that
 \begin{equation}\label{zeta}
 \zeta(x)=\frac{1}{\sqrt{8\pi}}e^{\frac{-|x|}{2}}
 \end{equation}
  is the ground state of the hydrogen atom in $\R^3$ 
 with its nucleus at zero, \cite{landau}. 
 For a simple proof of the fact that $\zeta$ is the (up to a constant) unique ground state of the hydrogen atom we refer to the lecture notes \cite{loss} Chapter 2. It follows that the ground state energy of hydrogen is
 \begin{equation}\label{Eh}
  E_h:=\inf\sigma(-\Delta_x-\frac{1}{|x|})=-\frac{1}{4}.
 \end{equation}
  In Section  \ref{abschaetzungfuerEe} we prove that $|E_{e^-}|=\frac{|E_h|}{16}$ so $|E_{e^-}|$ is one sixteenth of the (first) ionization energy of the hydrogen atom.  
 All experimentally measured  first ionization energies of atoms are bigger than $|E_{e^-}|=\frac{|E_h|}{16}$, namely bigger than a sixteenth of the (first) ionization energy of the hydrogen atom. We refer, for example,  to \cite{MarMK}. 
   Higher ionization energies $\inf\sigma(Q_{N-k} H_{N-k})- \inf\sigma(Q_{N-k+1} H_{N-k+1})$,  $k>1$ are expected to be bigger than the first ones and every experiment confirms this. As a consequence, we expect \eqref{molvermutung} to hold.

   Note that \eqref{molvermutung} is important 
not only for the proof of the main theorem, but already for  
proving that $H(r,v)$ has a ground state, which we do for $r$ large enough. In the case of the hydrogen atom  we do this in  
Section \ref{HVZ} below, by proving a Weyl type Theorem for $H(r,v)$ which helps us prove that its ground state energy is below the bottom of its essential spectrum. 
For non-experts: the fact that it is energetically favorable for the electron to stay close to the nuclei, at least when the nuclei are not too close to the  plate, helps to show compactness of energy minimizing sequences. 
This is in contrast with the fact that the information "close to the (unbounded) plate" would not clearly ensure compactness of energy minimizing sequences.  In Section \ref{generalizationmolecule} we prove that \eqref{molvermutung} implies existence of a ground state of the system for large enough distances in the general case. The proof is done in another way, see equation \eqref{existencegroundstate} below and its proof.

\subsection{The main result}

Let 
\begin{equation}\label{ENdef}
E_N=\inf \sigma( H_N|_{\text{Ran}Q_N}),
\end{equation} 
 be the ground state energy of $H_N|_{\text{Ran}Q_N}$ when the latter acts on the whole space and not only on the half--space. 
The HVZ (see e.g. \cite{Hunziker}, \cite{vanWinter}, \cite{Zhislin}) and Zhislin-Sigalov theorems (see e.g. \cite{Zhislin}, \cite{ZhislinSigalov}) imply that $E_N$ is an eigenvalue of $ H_N|_{\text{Ran}Q_N}$, lying strictly below the essential spectrum:
\begin{equation}\label{HVZZhislin}
E_N<\min\sigma_{\rm ess}\big( H_N|_{\text{Ran}Q_N}\big).%=E_{N-1}.
\end{equation}
The ground states of the Hamiltonian $H_N$ are exponentially decaying, (see e.g. \cite{CombesThomas}, \cite{griesemer}) namely there exists $c>0$ such that
\begin{equation}\label{eq:expdecay}
H_N \Phi= E_N \Phi \implies \|e^{c|x|} \partial^\alpha\Phi\|_{L^2} < \infty, \quad |\alpha|\leq 2.
\end{equation}
Assuming \eqref{molvermutung}   the interaction energy of the system of the molecule with the plate is defined by
\begin{equation}\label{interen}
W(r,v)=E(r,v)- E_N.
\end{equation}
If the interaction energy is negative, i.e.,  $E(r,v)<E_N$, 
then it is  attractive, because separating the molecule  from the plate costs energy. 
Similarly $E(r,v)> E_N$ implies a positive interaction energy and it is repulsive. While we are not able to estimate the interaction energy for small distances, we are able to study its long range asymptotics. 
%\textbf{Note that we have when the distance of one of the nuclei to the plate does to zero then the interaction energy goes to $-\infty$.
%	This is not surprising since for  This relies on the fact that for small distances of the nuclei to the wall the Born-Oppenheimer approximation is not consistent with the Dirichlet boundary condition.} 
Let 
\begin{equation}\label{def:B}
B:=\{\Psi \in L^2(\R^{3N}): \Psi \text{ is a ground state of } H_N|_{\text{Ran}Q_N}\}
\end{equation}
We define
\begin{equation}\label{def:Cv}
C(v)= \frac{1}{16} \sup_{\psi \in B, \|\psi\|=1}\left\langle \psi, \left(\left(\sum_{i=1}^N x_i \cdot v\right)^2 +\left|\sum_{i=1}^N x_i\right|^2\right)\psi\right\rangle.
\end{equation}
and 
\begin{equation}\label{def:Dv}
D= \frac{1}{4} \sup_{\psi \in M}\left\langle \psi,  \left(\sum_{i=1}^N x_i^4\right) \psi\right\rangle,
\end{equation}
where $M$ is the set of maximizers of the right hand side of \eqref{def:Cv}. For each ground state $\psi \in B$ we define the one electron density 
\begin{equation}
\rho_\psi(x)=\int |\psi(x,x_2,\dots,x_N)|^2 dx_2 \dots dx_N.
\end{equation}
Recall that $\beta$ is given by \eqref{dielectriccon}.
We are now ready to state our main result. 
\begin{Theorem}\label{satzallg}
(a)	Under the binding condition \eqref{molvermutung}, 
	\begin{equation}\label{eqn:mainthm}
	W(r,v)=-\beta\frac{C(v)}{r^3} +   \mathcal{O}\left(\frac{1}{r^4}\right), \quad \text{as } r\to\infty, 
	\end{equation}
	where $C(v)$ is given in  \eqref{def:Cv}.

(b) If moreover the molecule is symmetric with respect to the map $x \to -x$ and so is the one electron density  $\rho_\psi$ for all $\psi \in B$, then we have the more precise expansion 
	\begin{equation}\label{eqn:cor1}
	W(r,v)=-\beta\frac{C(v)}{r^3} +   \mathcal{O}\left(\frac{1}{r^5}\right), \quad \text{as } r\to\infty. 
	\end{equation}
	In the special case that the molecule is an atom and the one electron density $\rho_\psi$ is spherically symmetric for all $\psi \in B$, we have the more precise expansion
	 	\begin{equation}\label{eqn:cor2}
	 	W(r)=-\beta\frac{C}{r^3} - \beta\frac{D}{r^5} + \mathcal{O}\left(\frac{1}{r^6}\right), \quad \text{as } r\to\infty, 
	 	\end{equation}
	 	where we have omitted the orientation $v$, because there is no dependence on it. 
\end{Theorem}
\begin{remark}
	As it is obvious from the definition \eqref{def:Cv},   the constant $C(v)$ is positive for all $v$. Thus  \eqref{eqn:mainthm} implies attraction independently of the orientation and of the form of the ground states of the molecule.  Theorem \ref{satzallg} is the first rigorous result providing the leading term of a van der Waals interaction without any restrictions on the multiplicity of the ground state energy. This is something that does not rely on improvements of the existing methods of previous works but on the fact that the mirror charges are strictly correlated to the true ones unlike the case of interaction between atoms or molecules. 
\end{remark}
\begin{remark}
	In \cite{anapolitanoslewin} the van der Waals interaction between molecules was studied. The term of the order $\frac{1}{r^3}$ is the leading term when the molecules have dipole moments. Other than in Theorem \ref{satzallg}, whether it is attracting or repulsive depends on the orientations of the molecules. This is because the dipole-dipole interaction between molecules comes from  permanent dipole moments of the molecules and the orientations could be so that the molecules repel each other. 
\end{remark}
\begin{remark}
	Equation \eqref{eqn:cor1} can be of interest because there are a lot of molecules having the respective symmetry, including e.g. molecules consisting of two atoms of the same kind. 
\end{remark}

 We will start proving Theorem \ref{satzallg} for the special case of a hydrogen atom. In this case
our first main result is  more precise and moreover \eqref{molvermutung} does not need to be assumed. We also omit the unit vector $v$ as in this case the interaction energy does not depend on it. 
\begin{Theorem}
	\label{satz} Assume that $N=1$, $M=1$ (hydrogen atom). 
	There are $r_0>0$ and $D_1,D_2,D_3>0$, so that for all $r>r_0$
	$$ -\frac{D_3}{r^6} \leq W(r)+\frac{\beta}{r^3} +\frac{18\beta}{r^5} \leq D_1 e^{-D_2 r}\, .$$
\end{Theorem}
\begin{Bemerkung}
	It is worth noting that the  two leading terms of the interaction energy can be calculated explicitly. This is not the case, even for the leading order term, in the usual van der Waals law for a  system consisting of two hydrogen atoms. This relies  on the fact, mentioned above,  that in this regime the  mirror charges are strictly correlated to the true ones. 	
	 Our theorem shows that the interaction energy of the hydrogen atom with the plate is, in leading order, given by $-\frac{\beta}{r^3} -\frac{18\beta}{r^5}$. 
\end{Bemerkung}

For the proof of Theorems \ref{satzallg}, \ref{satz},  we use, to a large extent, the methods already employed in \cite{anapolitanossigal}  to investigate the van der Waals interaction between atoms. However, this approach has to be significantly modified, which is  mostly due to the fact that the operators act on a half--space. 
 Even showing that the system has  a ground state in the case of a general molecule interacting with a plate, is an open problem. Assuming \eqref{molvermutung}, we can prove existence of a ground state,  when $r$ is large enough,  and  Theorem  \ref{satzallg}. However, we are able to prove  \eqref{molvermutung} only in the special case of a hydrogen or a helium atom. As explained above we expect \eqref{molvermutung} to hold for a general molecule and there is a clear experimental evidence that it holds for each atom. 

 That the binding Condition \eqref{molvermutung} is necessary for \eqref{eqn:mainthm} can be physically seen as follows: if \eqref{molvermutung}  does not hold then the electrons prefer to organize themselves so that some of them are close to the plate and possibly some of them close to the nuclei. But then we have interaction of the positive ion with its mirror image which gives a Coulomb attraction. There is moreover, a dipole-charge attraction of the positive ion with the dipoles of electron--mirror electrons pairs which competes with a dipole-dipole repulsion of the electron--mirror electron pairs.   In Section \ref{proofnecessity} we sketch the proof of the necessity of \eqref{molvermutung}.

%----------------------------------------------------------------------------------------------------------------------------

%----------------------------------------------------------------------------------------------------------------------------

\section{Basic properties of the half--space system.}\label{preparations}

In order to be able to prove Theorem \ref{satz}, we need a few important results, that we discuss in this Section.

\subsection{The electron-plate system and proof of \eqref{molvermutung} for the cases of a hydrogen and of a helium atom}\label{abschaetzungfuerEe} In this section we 
compute  $E_{e^-}$, the minimal energy of the free electron in a half--space with a perfectly conducting boundary, and prove \eqref{molvermutung} for the cases of a hydrogen atom and of a helium atom. 
We start with the following Hardy type inequality. 
\begin{Lemma}
For $u \in H_0^1(\R_+)$ we have 
\be
\label{zweiteungleichung}
\int_0^\infty\frac{|u(y)|^2}{4y^2}dy \leq\int_0^\infty|u'(y)|^2dy.
\ee
\end{Lemma}
\begin{proof}
	The proof is known but we shall repeat it for convenience of the reader. 
 By a density argument it is enough to prove
 the inequality \eqref{zweiteungleichung} for all $u \in C_c^\infty(\R_+)$. Indeed, if $u \in C_c^\infty(\R_+)$, then we have that
  \begin{align*}
 	\int_0^\infty\frac{|u(y)|^2}{4y^2}dy&=-\int_0^\infty\frac{\overline{u(y)}u(y)}{4}\bigg(\frac{1}{y}\bigg)'dy=Re\int_0^\infty\frac{\overline{u(y)}}{2y}u'(y)dy\\
 	 \overset{Cauchy-Schwarz}{\leq} & \bigg(\int_0^\infty \frac{|u(y)|^2}{4y^2} dy\bigg)^\frac{1}{2}\bigg(\int_0^\infty|u'(y)|^2dy\bigg)^\frac{1}{2},
 \end{align*}
 from which \eqref{zweiteungleichung} immediately follows.
\end{proof}
We are next going to prove 
\begin{equation}\label{EeminusgleichEh4}
E_{e^-}=\frac{E_h}{16},
\end{equation}
  where recall that $E_h$, given in \eqref{Eh},
 is the ground state energy of the hydrogen atom.  This implies \eqref{molvermutung} in the case of hydrogen atom for any $\beta \in (0,1]$ since $E_h<0$. 
Split $\R_+^3:=\R_+ \times \R^2$, then 

\begin{align}\nonumber
	E_{e^-}&=\underset{\psi\in H_0^1(\R_+^3), \| \psi \|_{ L^{2}(\R_+^3)}=1}{\inf} \la\psi\mid H_{e^-}\psi\ra\\ \nonumber
	&=\underset{\psi\in H_0^1(\R_+^3), \| \psi \|_{ L^{2}(\R_+^3)}=1}{\inf} \la\psi\mid\Big(-\Delta_z-\frac{1}{2|z-z^*|}\Big) \psi\ra\\ \label{Eeminusgeq}
	& = \underset{\psi\in H_0^1(\R_+^3), \| \psi \|_{ L^{2}(\R_+^3)}=1}{\inf} \la\psi\mid\Big(-\frac{d^2}{dz_1^2}-\frac{1}{4z_1}\Big) \psi\ra,
\end{align}
where the last equality holds due to separation of variables and the known fact that $\sigma(-\frac{d^2}{dz_2^2}-\frac{d^2}{dz_3^2}) =[0,\infty)$. 
Observe now that  
\begin{equation}\label{Hegroundstate}
\frac{z_1 e^{-z_1/8}}{8\sqrt{2}}
\end{equation}
 is a positive eigenfunction of $-\frac{d^2}{dz_1^2}-\frac{1}{4z_1}$ and therefore by Perron-Frobenius theory (see e.g. \cite{Reed4}  Chapter XIII Section 12) it has to be its ground state. It follows that
 \begin{equation}\label{Eeequality}
  E_{e^-}=-\frac{1}{64},
  \end{equation}
  which together with  \eqref{Eh} implies \eqref{EeminusgleichEh4}. 

\medskip

From \eqref{Eh} and \eqref{Eeequality} we find that \eqref{molvermutung} holds in the case of the hydrogen atom.  We will now prove \eqref{molvermutung} for the case of a helium atom.

\begin{Proposition}[\eqref{molvermutung} for a helium atom]
	If $N=2$ and $M=1$ (two electrons and one nucleus), then \eqref{molvermutung} holds.
\end{Proposition}
\begin{proof}
	We may assume without loss of generality that $\beta=1$. 
	Note that because of the presence of spin, the spacial part of the two--electron wave function can be symmetric. 
	Thus we can use the tensor product of the rescaled hydrogen ground state 
	$\phi(x):=2^\frac{3}{2}\zeta(2x)$ with itself  as a test function for the ground state energy of helium. 
	Using Newton's Theorem, see Section 9.7 in \cite{newton}, 
	an elementary but lengthy computation gives 
	\begin{equation*}
	\la \phi \otimes \phi, H_2 \phi \otimes \phi\ra = 5.5 E_h,
	\end{equation*}
	therefore, 
	\begin{equation}\label{E2upper}
	\inf \sigma(H_2) \leq 5.5 E_h.
	\end{equation}
	Furthermore, by the rescaling argument in \eqref{res1} and \eqref{res2},  we have
	$\inf \sigma(H_1) = 4 E_h $. As a consequence, the binding condition \eqref{molvermutung} is clearly satisfied for $k=1,2$, i.e., 
	for the helium--plate system.
\end{proof}

\subsection{Boundedness from below of $H(r,v)$ and realization as a self-adjoint operator}
\label{infimumnichtunendlich}
Arguing similarly as in the proof of \eqref{zweiteungleichung},
one can prove that
\begin{align} \nonumber
& \iiint_{\R_+^3} \frac{|u(z_1,z_2,z_3)|^2}{4z_1^2}dz_1 dz_2 dz_3 \leq \\ \label{hardytype}  \iiint_{\R_+^3}& \left|\frac{\partial{u(z_1,z_2,z_3)}}{\partial z_1}\right|^2 dz_1 dz_2 dz_3, \text{ } \forall u \in H_0^1(\R_+^3).
\end{align}
Recall that  $\R^3_{r,v}$ was defined in \eqref{def:R3rv}.
Since for all $x=(x_1, x_2, x_3) \in \R^3_{r,v}$
\be \nonumber
\frac{1}{\vert 2re_1-(x-x^*)\vert}=\frac{1}{2\vert r-x_1\vert}=2\frac{1}{2}\frac{1}{2\vert r-x_1\vert}\leq\frac{1}{4\varepsilon}+\frac{\varepsilon}{4\vert r-x_1\vert^2}, \quad \forall \varepsilon > 0
\ee
using \eqref{hardytype} we find that in $H_0^1(\R^3_{r,v})$
\be \nonumber
\frac{1}{2\vert 2re_1-(x-x^*)\vert}\leq \frac{1}{8 \varepsilon}-\frac{\varepsilon}{2}\frac{\partial^2}{\partial x_1^2}  \leq \frac{1}{8\varepsilon}-\frac{\varepsilon}{2} \Delta, \quad \forall \varepsilon>0.
\ee
%\be \label{ineqdeux}
%\Rightarrow -\frac{1}{2}\Delta-\frac{1}{4\vert r+x_1\vert}\geq -\frac{1}{8}.
%\ee
This together with the infinitesimal boundedness of the Coulomb potential with respect to the Laplacian  implies that $H(r,v)$ is bounded from below in the case of a hydrogen atom. %By putting \eqref{inequn} and \eqref{ineqdeux} together we have:
%\be\label{preboundEr}
%-\Delta-\frac{1}{|x|}-\frac{1}{4\vert r+x_1\vert}\geq -\frac{17}{8}, \text{ in } H_0^1(\R^3_r).
%\ee
%Hence, for $\psi \in H_0^1(\R^3_r)$ with $\|\psi\|_{L^2(\R^3_r)}^2=1$,
%\be
%\begin{split}
	%\la\psi|H(r)\psi\ra&=\la\psi|\bigg%(-\Delta_{x}-\frac{1}{\vert %x\vert}+\halb\Big(-\frac{1}{2r}-\f%rac{1}{\vert %2re_1+(x-x^{*})\vert}+\frac{2}{\ve%rt %2re_1+x\vert}\Big)\bigg)\psi\ra\\
	%&\geq %\la\psi|\bigg(-\Delta_{x}-\frac{1}%{\vert x\vert}-\frac{1}{4\vert %r+x_1\vert}\bigg)\psi\ra-\frac{1}{%4r}\\ \nonumber
	%&\stackrel{\eqref{preboundEr}}{\ge%q} -\frac{17}{8}-\frac{1}{4r},
%\end{split}
%\ee
%proving that $E(r) > -\infty$. 

%So far we have considered $H(r,v)$ as a quadratic form. Note however, that because of \eqref{hardytype} and the Hardy inequality the quadratic form is closed on $H_0^1(\R^3_r)$. Since it is also bounded from below, by the KLMN Theorem (see e.g. \cite{Birman}, or  \cite{Reed} Theorem X.17) $H(r)$ can be realized as a self-adjoint operator with form domain $H_0^1(\R^3_r)$. 

We will explain how to prove that $H(r,v)$ is self--adjoint with domain $\mathcal{H}_{0,N} \cap H^2\left((\R^3_{r,v}\times\{\pm1/2\})^N,\C\right)  $,
where recall that $\mathcal{H}_{0,N}$ was defined in \eqref{formdomain}.
We start with the special case of a hydrogen atom.  Since the only difficulty arises from the attraction of the electron with its mirror image, we are going to prove that $H_{e-}$ defined in \eqref{Heminus} with domain $H^2(\R_+^3) \cap H_0^1(\R_+^3)$ is self-adjoint, where $\R_+^3:=\R_+ \times \R^2$.
The symmetry can be proven by integration by parts and an approximation by smooth functions. To prove self--adjointness
we first observe that since  $H_{e-}=-\Delta_z-\frac{1}{4z_1}$ from \eqref{hardytype}
we find that for all $u \in H^2(\R_+^3) \cap H_0^1(\R_+^3)$
\be
\int \frac{|u(z)|^2}{z_1^2}dz \leq \int 4 |\nabla u(z)|^2 dz = 4 \int u(z) (-\Delta u(z)) dz
$$$$ \leq 4 \|u\| \|-\Delta u\| \leq \epsilon \|-\Delta u\|^2 + \frac{4}{\epsilon} \|u\|^2.
\ee
Therefore, if we choose $\epsilon < 1$, the Kato-Rellich Theorem, see for example \cite{Reed} Theorem X.12, is applicable if we assume that $-\Delta$ is self-adjoint in $H^2(\R_+^3) \cap H_0^1(\R_+^3)$. That $-\Delta$ is self-adjoint is known by elliptic regularity since the plane is a smooth boundary. For convenience of the reader, however, we are going to provide a proof for this case, which is much simpler than the case of a general domain.   To do this we will use the basic criterion of self-adjointness according to which it suffices to prove that
$\text{Ran}(-\Delta + 1)=L^2(\R_+^3)$. 
Indeed, let $f$ be in $L^2(\R_+^3)$. Then its odd extension $\tilde{f}$ defined by
\begin{equation}
\tilde{f}(x_1,x_2,x_3)=
\begin{cases}
f(x_1,x_2,x_3), \text{ if   }x_1 \geq 0 \\
f(-x_1,x_2,x_3), \text{ if   }  x_1 < 0,
\end{cases}
\end{equation}
is in $L^2(\R^3)$. Thus
$g=(-\Delta + 1)^{-1} \tilde{f}$ is in $H^2(\R^3)$
and since it is odd as well it follows that $g(0,.,.)=0$, in the sense of a trace theorem which we state and prove for convenience of the reader in Appendix \ref{App:trace}. Using this we show in Appendix \ref{App:trace} that
  $$g|_{\R_+^3} \in H_0^1(\R_+^3)$$
and
$$(-\Delta +1)g|_{\R_+^3}= f,$$
 which completes the proof of self-adjointness of $-\Delta$ and therefore of $H_{e-}$ with operator domain $H^2(\R_+^3) \cap H_0^1(\R_+^3)$.	 

In the general case of a molecule self-adjointness and boundedness from below can be proven as  in the special case of hydrogen atom. The only terms that are new are the terms of electrons interacting with the mirror electrons of other electrons  in \eqref{eqI1}, but those can be controlled  with the arguments  of Lemma \ref{lemmatrapezoid}. The self-adjointness of $H(r,v)$ can also be proven in a similar fashion as in the case of the hydrogen atom-plate system. Note that the arguments of Appendix \ref{App:trace} are applicable because the boundary of $(\R_+ \times \R^2)^N$ consists of $N$ hyperplanes in $\R^{3N}$.

\subsection{The Feshbach map}
As in \cite{anapolitanossigal}, a main ingredient of the proof is the Feshbach map, which we now introduce. 

\begin{Definition}
\label{feshbach}
Let $\mathcal{H}$ be a separable Hilbert space,  $H$ a self-adjoint operator on $\mathcal{H}$ with domain $D(H)$, $P$  an orthogonal projection of finite rank with $\text{Ran}P \subset D(H)$, and $H^{\bot}=P^{\bot}HP^{\bot}$. For $\lambda \in \R$, so that $H^{\bot}-\lambda$ invertible, the Feshbach map is defined as 
$$F_P(\lambda)=PHP-PHP^{\bot}(H^{\bot}-\lambda)^{-1}P^{\bot}HP\mid _{Ran P}.$$

\end{Definition}

The following theorem is well--known, see \cite{feschbach} for example.
\begin{Theorem}
\label{TheoremF}
Let $H$, $\mathcal{H}$, $P$ be as above. For $\lambda \in \R$, we assume that 
\be H^{\bot}-\lambda \geq c \ee for some $c>0$. Then 
\be\label{Feshbacheigen} \lambda \text{ is an eigenvalue of } H \Leftrightarrow \lambda \text{ is an eigenvalue of } F_P(\lambda).
\ee

\end{Theorem}

\begin{Bemerkung}
Theorem \ref{TheoremF} will play a central role in our proof of Theorem \ref{satz}. We will mostly need the implication "$\implies$" in \eqref{Feshbacheigen}. More precisely, we will use: if $E$ is the ground state energy of $H$ which is  strictly below the essential spectrum and 
\be  \label{conditionF} H^{\bot}-E \geq c >0, \ee
then $E$ is an eigenvalue of $F_P(E)$. In the special case that P has rank one, $F_P(E)$ can be identified with a scalar and using this identification we obtain 
\begin{equation}\label{Feshbachscalar}
E=F_P(E).
\end{equation}
 
\end{Bemerkung}

In the next section we prove that Condition \eqref{conditionF} is satisfied when the hydrogen atom is far enough from the plate, so Theorem \ref{TheoremF} can be applied. In Section \ref{generalizationmolecule}, we will do this for a general molecule.

%_____________________________BEDUNGUNG___________________________________________________
\subsection{Proof of the lower bound \eqref{conditionF} in the case of a hydrogen atom for a suitable projection $P$}\label{Beweis der Bedingung}
% We are giving the 

\begin{proof}%[Proof of condition \eqref{conditionF}]
We closely follow some ideas in  \cite{anapolitanossigal}.
First, one has to find a suitable projection $P$.

 We consider a  spherically symmetric $C^\infty$-function $h$ with $0 \leq h \leq 1$, with support in $B(0,\frac{1}{4})$, where $h=1$ in $B(0,\frac{1}{5})$ and let $h_r(x)=h(\frac{x}{r})$. We set
\be \label{abgeschnitten}
\psi(x):=\frac{h_r(x)\zeta(x)}{\|h_r(x)\zeta(x)\|},
\ee
where recall that $\zeta$, defined in \eqref{zeta}, is the ground state of the hydrogen atom. 
In other words $\psi$ is a cutoff ground state of the hydrogen atom. The presence of the cutoff function $h_r$ ensures that $\psi$ is in the domain of $H(r)$, recall \eqref{defHlin}.

Lastly, we define the projection 
\begin{equation}\label{Pdef}
P:=P_{\psi}=\mid \psi\hspace{1mm}\ra \la\hspace{1mm} \psi\mid. 
\end{equation}

To show that Condition \eqref{conditionF} is satisfied, we use the IMS localization formula, see for example \cite{cycon} Chapter 3.1. In the form that we need it the formula reads
\be\label{IMSH} H(r)=J_1H(r)J_1+J_2H(r)J_2-\vert\nabla J_1\vert^2-\vert\nabla J_2\vert^2\ee
where $J_1,J_2:\R^3\rightarrow\R$ are two $C^\infty$-functions, that have bounded derivatives and satisfy the equality $J_1^2+J_2^2=1$.

The IMS localization formula is a very helpful tool, since it allows us to evaluate the Hamiltonian in two different subspaces, near to the nucleus and far from it. This is easier to do than analyzing the Hamiltonian in the whole space directly. The price to pay however is the localization error $\vert\nabla J_1\vert^2+\vert\nabla J_2\vert^2$, which one must take into account.\\

First we construct the functions $J_1$ and $J_2$. Let $\chi_1,\chi_2:\R^3\rightarrow\R$ be two functions with $\chi_1,\chi_2\in C^{\infty}$, $0\leq\chi_i\leq1$ for $i \in \{1;2\}$, given by

\be \label{f1}
\chi_1(y) =\begin{cases} 0, &\vert y\vert\leq\frac{1}{4} \\
                     1,& \vert y\vert\geq\frac{2}{7}
       \end{cases}
\ee
and

\be \label{f2}
\chi_2(y) =\begin{cases} 1, &\vert y\vert\leq\frac{2}{7} \\
                     0,& \vert y\vert\geq\frac{1}{3}.
       \end{cases}
\ee

We define the functions $J_1$ und $J_2$ as follows:

\begin{align} \label{J1def}& J_1(x)=\frac{\chi_1(\frac{\vert x\vert }{r})}{\sqrt{\chi_1(\frac{\vert x\vert }{r})^2+\chi_2(\frac{\vert x\vert }{r})^2}}\\
\label{J2def}&J_2(x)=\frac{\chi_2(\frac{\vert x\vert }{r})}{\sqrt{\chi_1(\frac{\vert x\vert }{r})^2+\chi_2(\frac{\vert x\vert }{r})^2}}
\end{align}
\begin{center}
\includegraphics[scale=0.4]{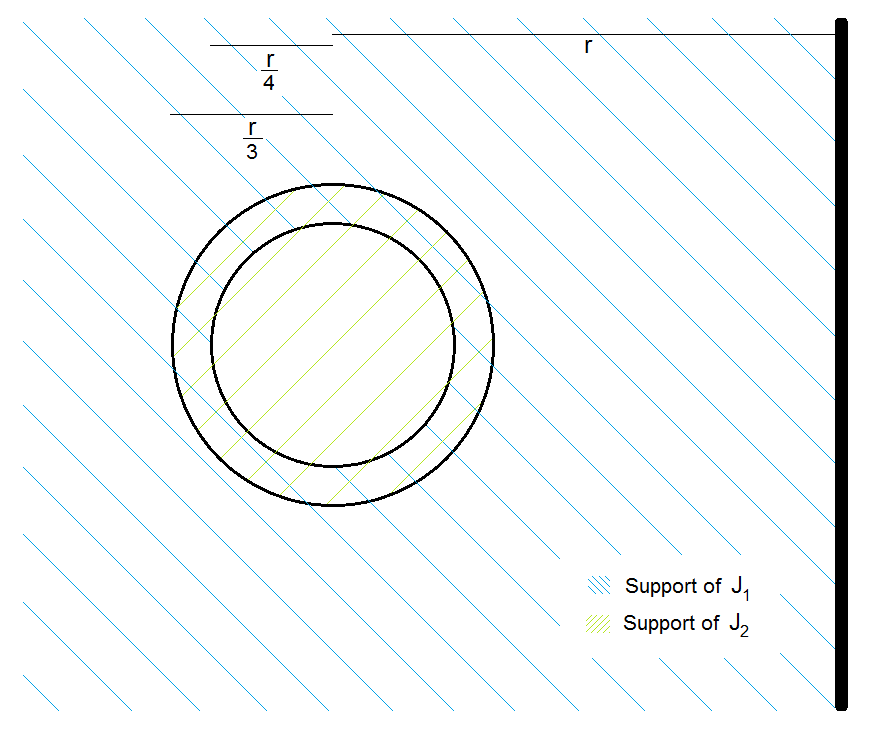}\\
\small Figure 4
\end{center}
  Figure 4   has a sketch of the supports of the functions $J_1, J_2$. Since $\chi_1^2 + \chi_2^2 \geq 1$, the functions $J_1, J_2$ inherit the $C^\infty$ property of $\chi_1$ and $\chi_2$ and moreover
 \begin{equation}\label{partition}
 J_1^2+J_2^2=1.
 \end{equation}
  Their derivatives are compactly supported and therefore bounded. Consequently, the IMS localization formula is applicable. On the support of $J_1$ the electron is far from the nucleus and on the support of $J_2$ the electron is close to the nucleus. Observe also that due to \eqref{J1def}, \eqref{J2def} there exists  $D>0$ such that
\begin{equation}\label{nablaJ12}
|\nabla J_1|^2 + |\nabla J_2|^2 \leq \frac{D}{r^2}.
\end{equation}
Moreover, due to \eqref{abgeschnitten} and the choice of the functions $\chi_1, \chi_2$ we have on $\text{supp} \psi$ that $J_2=1$ and $J_1=0$. Consequently,
\begin{equation}\label{Jpcom}
P J_i= J_i P, \quad i=1,2.
\end{equation}
We assess in the following the terms of the IMS localization formula (\ref{IMSH}) individually. First, we look at $J_1 H(r) J_1$.  Using \eqref{defHlin} we find 
\begin{equation*}
J_1 H(r) J_1 \geq J_1\bigg(-\Delta_{x}-\frac{1}{\vert x\vert}-\frac{\beta}{4r}-\frac{\beta}{2\vert2re_1+(x-x^{*})\vert}\bigg)J_1 
\end{equation*}
and therefore using that $\frac{\vert x\vert }{r} \geq \frac{1}{4}$ on $\text{supp} J_1$ we find that
\begin{align}
 J_1 H(r) J_1  \geq J_1\bigg(\underbrace{-\Delta_{x}-\frac{\beta}{2\vert 2re_1+(x-x^{*})\vert}}_{\substack{\geq \beta^2 E_{e^{-}}}}-\frac{C_1}{r}\bigg)J_1 \\ \label{J1H}
\geq \bigg( \beta^2 E_{e^{-}}-\frac{C_1}{r}\bigg)J^{2}_1,
\end{align}
for some $C_1>0$, where in the last step we used \eqref{Eescaling}. 
Therefore, 
\be
 P^{\bot}J_1 H(r) J_1 P^{\bot}= \bigg( \beta^2 E_{e^{-}}-\frac{C_1}{r}\bigg)P^\bot J^{2}_1 P^\bot.
\ee
\smallskip

\noindent
Doing the same for $J_2 H(r) J_2$ and using that $\frac{\vert x\vert }{r}\leq \frac{1}{3}$ on $\text{supp} J_2$ we extract 
\begin{align}\nonumber
J_2 H(r) J_2&\geq J_2\bigg(-\Delta_{x}-\frac{1}{\vert x\vert}-\frac{\beta}{4r}-\frac{\beta}{2\vert 2re_1+(x-x^{*})\vert}\bigg)J_2
%\\ \nonumber
%&\geq J_2\bigg(-\Delta_{x}-\frac{1}{\vert %x\vert}-\frac{\beta}{4r}-\frac{3\beta}{8r}\bigg)J_2
\\ \nonumber
&\geq J_2\bigg(-\Delta_{x}-\frac{1}{\vert x\vert}-\frac{C_2}{r}\bigg)J_2,
\end{align}
 for a suitable constant  $C_2>0$.
Thus 
\begin{align}\nonumber
P^{\bot}J_2 H(r) J_2P^{\bot}&\geq P^{\bot}J_2\bigg(-\Delta_{x}-\frac{1}{\vert x\vert}\bigg)J_2P^{\bot}-P^\bot \frac{C_2}{r}J^{2}_2 P^\bot \\
\geq  \Big(E_h+& d\Big)  P^\bot J^{2}_2P^\bot -\frac{C_2}{r}P^\bot J^{2}_2 P^\bot \mbox{  }\mbox{  }\mbox{  }\mbox{  }\mbox{  with }d>0\label{J2H},
\end{align}
where in the last step we used \eqref{Jpcom}. The gap $d$ originates from the fact that $P^\bot$ projects out of the ground state energy eigenspace of the hydrogen atom. 

Combining \eqref{IMSH}, \eqref{nablaJ12}, \eqref{J1H}, and  \eqref{J2H} we find 
\begin{align*}
H^{\bot}(r) &\geq\Big( \beta^2 E_{e^{-}}-\frac{C_1}{r}\Big)P^\bot J^{2}_1P^\bot + \Big(E_h+d-\frac{C_2}{r}\Big)P^\bot J_2^{2}P^\bot -\frac{D}{r^{2}}\,,
\end{align*}
  which together with \eqref{Jpcom} implies 
\begin{align*}
&H^\bot(r) \geq \Big( \beta^2 E_{e^{-}}-\frac{C_1}{r}\Big) J^{2}_1P^\bot + \Big(E_h+d-\frac{C_2}{r}\Big)J_2^{2}P^\bot -\frac{D}{r^{2}}.
\end{align*}
Let $\tilde{d}:=\min\{d;\beta^2 E_{e^-}-E_h\}$. Due to \eqref{Eeequality} and \eqref{Eh} we have that $\beta^2 E_{e^-}-E_h>0$ for all $\beta \in (0,1]$ and in particular $\tilde{d}>0$.  As a consequence, using also \eqref{partition} we obtain that there exists $C>0$ such that
\begin{align}\nonumber
H^\bot(r) & \geq\Big( E_h+\tilde{d}\Big)P^\bot-\frac{C}{r}P^\bot-\frac{D}{r^{2}}\\ \label{HbotEhbound}
&\geq  E_h+\tilde{d}-\frac{C}{r}-\frac{D}{r^{2}}, 
\end{align}
where in the last step we used the inequalities  $P^\bot\leq 1$ and  $E_h + \tilde{d} \leq \beta^2 E_{e^-} <0$, see \eqref{Eeequality}. Recall that $E(r)$ is defined in \eqref{defE}, where here we do not write $E(r,v)$ as there is no dependence on $v$.
By the variational principle,  
\begin{equation}\label{Erroughbound}
E(r) \leq \la \psi, H (r)\psi \ra,
\end{equation}
where $\psi$ is the cutoff ground state of the hydrogen atom defined in \eqref{abgeschnitten}.
From \eqref{defHlin} and Lemma \ref{lemmatrapezoid} we find 
\begin{equation}\label{Hrlesshydrogen}
H(r) \leq -\Delta_x-\frac{1}{|x|}.
\end{equation}
Hence, using \eqref{zeta}, \eqref{abgeschnitten}, \eqref{Erroughbound} and \eqref{Hrlesshydrogen}   we obtain that there exists $c>0$ such that
\be \label{roughupperbound} 
E(r) \leq \la \psi, \left(-\Delta - \frac{1}{|x|}\right) \psi \ra = E_h + \mathcal{O}(e^{-cr}).
\ee
Using \eqref{HbotEhbound} and \eqref{roughupperbound}
it follows that if $r$ is not  small, then 
\begin{equation*}
	H^{\bot}(r)-E(r)\geq \frac{\tilde{d}}{2}>0\, .\qedhere
\end{equation*}
\end{proof}
In particular, we can use the Feshbach map to estimate the interaction energy of the hydrogen-plate system.

\subsection{Existence of a ground state: a  Weyl type Theorem in a half--space}\label {HVZ}
We already showed that the energy of the hydrogen-plate system is bounded from below, i.e.,  $E(r)> -\infty$. 
Now we want to prove existence of a ground state at least when the distance $r$ of the nucleus to the wall is not too small. We will do this by proving 
\begin{equation*}
	E(r) <\inf \sigma_{ess}(H(r))\, ,
\end{equation*}
i.e., the ground state energy is strictly below the essential spectrum, thus a ground state exists.

Physical intuition shows that the essential spectrum of a quantum system begins when the electron can escape to infinity.  In the general many-electron case
this is the content of the famous Hunziker-van Winter-Zhislin (HVZ) theorem, at least in the traditional case where a half--space is not present, see for example the original references \cite{Hunziker}, \cite{vanWinter}, and the beautiful proof of \cite{Zhislin}. In the one electron case this goes back to Weyl. 

Compared to the case of a free hydrogen atom our situation has two significant differences:  On the one hand, the electron does not need to have positive energy in order to be able to escape, but only an energy bigger than $\beta^2 E_{e^-}$, which is negative, since it can move along the surface of the plate to infinity. 
On the other hand, the nucleus alone interacting with the plate also has a negative energy, more precisely $-\beta/(4r)$ because of the presence of the plate. 

Hence one is lead to the conjecture  
\begin{equation}\label{infwesspek}
\inf \sigma_{ess}(H(r))=\beta^2 E_{e^-}-\frac{\beta}{4r},
\end{equation}
which we are going to prove in this section. 
  Observe also that \eqref{roughupperbound}
 together with \eqref{Eh} and \eqref{Eeequality} gives 
\be \label{Ergap}
E(r) < \beta^2 E_{e^-}-\frac{\beta}{4r},
\ee
if $r$ is not too small. As a consequence, it is enough to prove \eqref{infwesspek} in order to guarantee that the hydrogen-plate system has a ground state below the essential spectrum, 
at least for large enough $r$.

%To this end, we will need the following definition and lemma.

%\begin{Definition}
%	Let  $A$ be a self-adjoint operator acting on a Hilbert space $\mathcal{H}$ and $\lambda \in \R$. The sequence $(\psi_n)_{n \in \N}$ is called \underline{Weyl sequence} for $A$ and $\lambda$, if it satisfies the following conditions:\\
%	1. $\|\psi_n\|=1$ for all $n$.\\
%	2. $\|(A-\lambda)\psi_n\|\rightarrow 0$ for $n\rightarrow \infty$.\\
%	3. $\psi_n \rightarrow 0$ weakly when $n\rightarrow \infty$, namely $\la \phi|\psi_n\ra\rightarrow0$ for all $\phi \in \mathcal{H}$.
	
%\end{Definition}

%\begin{Lemma}\label{lemma}
%	Let  $A$ be a self-adjoint operator and $\lambda \in \R$. Then $\lambda \in \sigma_{\text{ess}}(A)$ if and only if there exists a Weyl sequence for $A$ and $\lambda$.
%\end{Lemma}
%For the proof of this lemma, see, for example, Theorem 23.55 in 
%\cite{mathematical concepts} or \cite{simon-geometric methods}.

\begin{Proposition}
 The equation \eqref{infwesspek} holds.
\end{Proposition}

\begin{proof}
	This proof is an adaptation of the proof of the classical HVZ theorem, see e.g. \cite{mathematical concepts} Section 12.4.	
	First we show that 
	\begin{equation}\label{infwesspekleq}
	\beta^2 E_{e^-}-\frac{\beta}{4r}\geq\inf \sigma_{ess}(H(r)).
	\end{equation}
	To do so it suffices to prove that
	\begin{equation}\label{preinfwesspekleq}
	[\beta^2 E_{e^-}-\frac{\beta}{4r},\infty) \subset \sigma(H(r)).
	\end{equation}
	Let $\lambda \geq \beta^2	E_{e^-}-\frac{\beta}{4r}$ and 
	\begin{equation*}
		\psi_r(x_1):= \beta^\frac{3}{2} \frac{\beta(x_1+r)e^{\frac{\beta(x_1+r)}{8}}}{8\sqrt{2}}
	\end{equation*}
	Then we have that 
	\begin{align}\label{HeminusGS}
	\left(-\frac{d^2}{dx_1^2}-\frac{\beta}{2|2re_1+(x-x^*)|}\right)	\psi_r(x_1) = \beta^2 E_{e^-} \psi_r(x_1).
	\end{align}
	One can see that either with a direct computation or by recalling that \eqref{Hegroundstate} is the ground state of $-\frac{d^2}{dz_1^2}-\frac{1}{4z_1}$, see also \eqref{Heminus}-\eqref{Eenegativ}, and by using the rescaling argument \eqref{Skalierung}-\eqref{Eescaling}.
	
	Since $\lambda -  \beta^2 E_{e^-} + \frac{\beta}{4r}\ge 0$, hence is in the spectrum of  $-\frac{d^2}{dx_2^2}-\frac{d^2}{dx_3^2}$, we can choose, for every $n \in \N$,  $\phi_n(x_2,x_3)$ with $\phi_n\in C_c^\infty(\R^2)$, $\|\phi_n\|_{L^2}=1$, \begin{equation}\label{supppsinfar}
	\text{supp} \phi_n \subset\{(x_2,x_3) \in \R^2: |(x_2,x_3)| \geq n\},
	\end{equation} 
	and
	\begin{equation}\label{phin}
	\left\|\left(-\frac{d^2}{dx_2^2}-\frac{d^2}{dx_3^2}-\Big(\lambda - \beta^2 E_{e^-} + \frac{\beta}{4r}\Big)\right) \phi_n\right\|  \leq \frac{1}{n}.
	\end{equation}
	Defining now $\varphi_n(x)=\psi_r(x_1)\phi_n(x_2,x_3)$ we get
 	$\|\varphi_n\|_{L^2}=1$ and using \eqref{HeminusGS}, \eqref{supppsinfar}, \eqref{phin} and  \eqref{defHlin} one finds 
	\begin{equation}\label{HeapproxGS}
	\left\| (H(r)-\lambda) \varphi_n \right\| \to 0.
	\end{equation}
	Therefore, we have that
	$\lambda \in \sigma(H(r)) $. Thus, \eqref{preinfwesspekleq} is true, which implies \eqref{infwesspekleq}.
	\smallskip
	
	Now we show 
	\begin{equation}\label{infwesspekgeq}
	\beta^2 E_{e^-}-\frac{\beta}{4r}\leq\inf \sigma_{ess}(H(r)).
	\end{equation}
	Consider $\lambda \in \sigma_{ess}(H(r)) $. Then there is a Weyl sequence $(\psi_n)_{n \in \N}$ with $\la \psi_n| H(r) \psi_n\ra \rightarrow \lambda$ for $n \rightarrow \infty$.\\
	Let $J_{1,R}, J_{2,R}$ be defined as in \eqref{J1def}, respectively \eqref{J2def}, but with $r$ replaced by a parameter $R$. Note that multiplication with $J_{j,R}$ leaves the domain $H_2(\R^3_r) \cap H_0^1(\R^3_r)$  of $H(r)$ invariant even though $J_{j,R}$ might have support out of the half--space. 
	One can observe using \eqref{defHlin} that
	\begin{align}\nonumber
	J_{1,R} H(r) J_{1,R} \geq J_{1,R}\left(-\Delta-\frac{\beta}{2|2re_1+(x-x^*)|}-\frac{\beta}{4r}\right) J_{1,R} & + \mathcal{O}\left(\frac{1}{R}\right)
	\\ \label{J1HJ1} \geq \left(\beta^2 E_{e^-}-\frac{\beta}{4r}\right) J_{1,R}^2 + \mathcal{O}\left(\frac{1}{R}\right)&,
	\end{align}
	where in the last step we used \eqref{Eescaling}. 
	Using the IMS localization formula for  $J_{1,R}, J_{2,R}$
	we can write
	\be \begin{split}
		H(r)&=J_{1,R} H(r) J_{1,R}+\underbrace{J_{2,R} H(r) J_{2,R}}_{\substack{\geq E(r) J_{2,R}^2}}\underbrace{-\vert \nabla J_{1,R}\vert^2-\vert \nabla J_{2,R}\vert^2}_{\substack{ \mathcal{O}(\frac{1}{R^2})}}\\
		&\stackrel{\eqref{J1HJ1}}{\geq} \left( \beta^2 E_{e^-}-\frac{\beta}{4r}\right) J_{1,R}^2+ E(r) J_{2,R}^2+ \mathcal{O}\left(\frac{1}{R}\right)\\
		&= \beta^2 E_{e^-}-\frac{\beta}{4r} + \left( E(r) + \frac{\beta}{4r}  -  \beta^2 E_{e^-} \right) J_{2,R}^2+ \mathcal{O}\left(\frac{1}{R}\right),
	\end{split}\ee
	where in the last equality we used that $J_{1,R}^2 + J_{2,R}^2=1$.
	It follows that
	\begin{align}\nonumber
	\la \psi_n| H(r) \psi_n\ra  \geq \beta^2 E_{e^-}  - \frac{\beta}{4r}  + \left(\frac{\beta}{4r} - \beta^2 E_{e^-}+E(r)\right) \|J_{2,R}\psi_n\|^2 + \mathcal{O}(\frac{1}{R}).
	\end{align}
	Since $J_{2,R}(H(r)+i)^{-1}$ is compact and because $(H(r)+i)\psi_n=(H(r)-\lambda)\psi_n+(i-\lambda)\psi_n \rightarrow 0$ weakly for $n\rightarrow \infty$ by the properties of the Weyl sequence $\psi_n$, we can conclude:
	\be
	J_{2,R}\psi_n=J_{2,R}(H(r)+i)^{-1}(H(r)+i)\psi_n\rightarrow0 \mbox{ for } n\rightarrow \infty .
	\ee
	Thus 
	\be
	\lambda = \lim_{n\rightarrow \infty}\la \psi_n| H(r) \psi_n\ra\geq \beta^2 E_{e^-} - \frac{\beta}{4r} + \mathcal{O}(\frac{1}{R}),
	\ee
	from which \eqref{infwesspekgeq} follows if we take the limit $R \to \infty$.
	From \eqref{infwesspekgeq} and \eqref{infwesspekleq} we obtain \eqref{infwesspek}. 
\end{proof}

From \eqref{infwesspek} and \eqref{Ergap} the existence of a ground state of $H(r)$ follows immediately.

%----------------------------------------------------------------------------------------------------------------------------
%-------------------------------------------  -------------------------------------------------------
%----------------------------------------------------------------------------------------------------------------------------

\section{Proof of the van der Waals asymptotic for the hydrogen-plate system}
\label{Beweis}
%_____________________________WECHSELWIRKUNGSENERGIE_________________________________________
We will use the following \\  
\textbf{Notation}: 
We say that $f(r)=\mathcal{O}(r^{-n})$ if there exists $C,D>0$ such that if $r \geq C$ then $|f(r)| \leq D r^{-n}$. If $f(r) \in L^2$ or $f$ is an operator in $L^2$ then the inequality is understood in terms of the $L^2$ norm or the operator norm, respectively, depending on the context.

\begin{proof}[Proof of Theorem \ref{satz}: ]
We follow the general  strategy of \cite{anapolitanossigal}.
First we will begin with a proof of a weaker result, which can be extended to the case of a molecule interacting with the plate. Then we will refine the strategy using properties of the hydrogen atom in order to obtain Theorem \ref{satz} and this part can be extended for proving \eqref{eqn:cor2}.

We have already shown in the previous Section that the assumptions of Theorem \ref{TheoremF} are satisfied if we take the projection operator $P$ as in \eqref{Pdef}. Thus, abbreviating $E=E(r), H=H(r)$, Theorem \ref{TheoremF} yields
\be E\psi=F_P(E)\psi\, ,\ee
hence
\begin{align}\nonumber
E=\la\psi|E\psi\ra&=\la\psi|(PHP -PHP^{\bot}(H^{\bot}-E)^{-1}P^{\bot}HP)\psi\ra\\ 
\label{formelE}
&=\la\psi|H\psi\ra-\la\psi|PHP^{\bot}(H^{\bot}-E)^{-1}P^{\bot}HP\psi\ra.
\end{align}
We start by estimating the first term of the right hand side of \eqref{formelE}, namely $\la\psi|H\psi\ra$ and show afterwards, that the second term is small.
%We have 
%\begin{equation}\label{Hdecomposition}
%H=H_h+\halb I, 
%\end{equation}
%where $H_h=-\Delta_x-\frac{1}{|x|}$ is the Hamiltonian of %the hydrogen atom and
%\be\label{Idef} I:=-\frac{1}{2r}-\frac{1}{\vert %2re_1+(x-x^*)\vert}+\frac{2}{\vert 2re_1+x\vert}.\ee
Inserting the Taylor expansion 
\be \nonumber
\frac{1}{|2re_1-u|}=\frac{1}{2r}+\frac{e_1 \cdot u}{4r^2}+\frac{3\big(e_1 \cdot u\big)^2-|u|^2}{16r^3}+\frac{5(e_1 \cdot u)^3-3(e_1 \cdot u) |u|^2 }{32|r|^4}+\mathcal{O}(r^{-5})
\ee 
in \eqref{defU}
and using that $e_1 \cdot x=-e_1 \cdot x^*$ we find 
\be \label{Iexpansion}
U(x)=\frac{-(x \cdot e_1)^2-|x|^2}{8r^3}+\frac{f_{odd}(x)}{16r^4}+ \mathcal{O}\bigg(\frac{|x^4|}{r^5}\bigg),
\ee
on  the support of $\psi$ where
\begin{equation}\label{foddineq}
f_{odd}(x)=15(e_1 \cdot x)^3 + 3(e_1 \cdot x) (|x|^2-|2 x \cdot e_1|^2).
\end{equation}
Inserting \eqref{Iexpansion} in \eqref{defHlin} and using the equality in \eqref{roughupperbound} we obtain 
\begin{align}\label{psiHeq}
\la\psi|H\psi\ra&= E_h+ \mathcal{O}(e^{-cr})+\la\psi|\halb U\psi\ra\\ \nonumber
&=E_h+ \mathcal{O}(e^{-cr})- \beta \int \frac{(x \cdot e_1)^2+|x|^2}{16r^3}|\psi(x)|^2dx\\ \nonumber
&+ \beta \int \frac{f_{odd}(x)}{32r^4}|\psi(x)|^2dx+ \mathcal{O}\bigg(\frac{1}{r^5}\bigg).
\end{align}
 Therefore, since $f_{odd}$ is an odd function of $x$, we find
\begin{align}\label{psiHpsi}
\la\psi|H\psi\ra=E_h - \beta \int \frac{(x\cdot e_1)^2+|x|^2}{16r^3}|\psi(x)|^2dx+  \mathcal{O}\bigg(\frac{1}{r^5}\bigg).
\end{align}
Since $\zeta(x)=\frac{1}{\sqrt{8\pi}}e^{\frac{-|x|}{2}}$
we have  $\int \frac{(x \cdot e_1)^2+|x|^2}{16}|\zeta(x)|^2dx=1$, which together with \eqref{abgeschnitten} and \eqref{psiHpsi} 
yields  
\begin{equation}
\\ \label{Erlowerbound}
\la\psi|H\psi\ra=E_h-\frac{\beta}{r^3}+\mathcal{O}(r^{-5}).
\end{equation}
Now we discuss the second term in (\ref{formelE}),  
\be
\la\psi|PHP^{\bot}(H^{\bot}-E)^{-1}P^{\bot}HP\psi\ra\, .
\ee
Clearly,  
\be \begin{split}
\label{PHPPHP}
\|PHP^{\bot}(H^{\bot}-E)^{-1}P^{\bot}HP\|&\leq\|(H^{\bot}-E)^{-1}\|\|P^{\bot}HP\|^2\\
&\leq\frac{1}{c}\hspace{1mm}\|P^{\bot}HP\|^2
\end{split}\ee
because $\|(H^{\bot}-E)^{-1}\|\leq\frac{1}{c}$ due to \eqref{conditionF}. So we need a bound on  $\|P^{\bot}HP\|$. Since $P^{\bot} P=0$, we find
\be \label{PbotHPdec}
\|P^{\bot}HP\|=\|P^{\bot}(H-E_h)P\|=\|P^{\bot}(H-E_h)\psi \|\, ,
\ee
and  using \eqref{defHlin} we derive
\be 
\|P^\bot H P\| \leq \|P^{\bot}\left(-\Delta_x-\frac{1}{|x|}-E_h\right)\psi\| + \|P^{\bot} \halb U \psi\|.
\ee
Using this together with the equality in \eqref{roughupperbound} we extract
\begin{equation}\label{HhminusEhpsiineq}
\|P^\bot H P\| \leq  \|P^{\bot} \halb U \psi\| + \mathcal{O}(e^{-dr}).
\end{equation}
Using moreover \eqref{Iexpansion}, \eqref{foddineq}, \eqref{abgeschnitten} and \eqref{zeta} we find that
\begin{equation}\label{Ipsiineq}
\|U \psi\| = \mathcal{O}(r^{-3}).
\end{equation}
From \eqref{PbotHPdec}, \eqref{HhminusEhpsiineq}, and \eqref{Ipsiineq} we extract that
\be \|P^\bot H P\| = \mathcal{O}(r^{-3}). \ee
By inserting this result in (\ref{PHPPHP}) we see that\\ 
\be\label{Feshbachnichtlinearabsch}
\|PHP^{\bot}(H^{\bot}-E)^{-1}P^{\bot}HP\|= \mathcal{O}(r^{-6}).
\ee
%---------------------------Alles zusammen---------------------------------------------------------------------------

Using \eqref{formelE}, \eqref{Erlowerbound},      \eqref{Feshbachnichtlinearabsch} we arrive at
\be
E(r)=E_h-\frac{\beta}{r^3} + \mathcal{O}(r^{-5}),
\ee
which together with \eqref{interen} implies 
\be\label{eq:presatz}
W(r)=-\frac{\beta}{r^3} + \mathcal{O}(r^{-5}).
\ee
This strategy gives us a weaker result than Theorem \ref{satz} but we can generalize it to the case of a molecule interacting with a plate.
\medskip

\noindent
 Using properties of the hydrogen atom, we will refine this strategy in order to obtain Theorem \ref{satz}. Observing that the subtracted term in \eqref{formelE} is positive and using \eqref{Feshbachnichtlinearabsch} we find that there exists $C>0$ such that
$$ -\frac{C}{r^6} \leq E(r) - \la\psi|H\psi\ra \leq 0, $$
which together with \eqref{interen} and \eqref{psiHeq} implies that
\begin{equation}\label{Wrest}
 -\frac{C}{r^6}  \leq W(r) - \beta \la\psi|\frac{U}{2}\psi\ra \leq \mathcal{O}(e^{-cr}). 
\end{equation}
Thus, if we manage to prove that there exists $D>0$ such that
\begin{equation}\label{psiIpsiest}
-\frac{1}{r^3}-\frac{18}{r^5} -\frac{D}{r^7}  \leq \la\psi|\frac{U}{2}\psi\ra \leq -\frac{1}{r^3}-\frac{18}{r^5} + \mathcal{O}(e^{-cr})\, ,
\end{equation}
then Theorem \ref{satz} follows immediately. In the rest of the section we will prove 
\eqref{psiIpsiest}. With the help of Newton's Theorem, see Section 9.7 in \cite{newton},  one sees 
\begin{equation}
 \int \frac{2}{\vert 2re_1-x\vert}|\psi(x)|^2 dx = \frac{1}{r},
 \end{equation}
which together with \eqref{defU} and the fact that $|2re_1-(x-x^*)|=2(r-x_1)$, where $x_1=x \cdot e_1$, gives
\be 
\la\psi|\frac{U}{2}\psi\ra = \frac{1}{4} \la \psi| \left( \frac{1}{r}-\frac{1}{r-x_1} \right) \psi\ra. 
\ee
However, we can rewrite,  
\be\label{geometricdecomposition}
\frac{1}{r}-\frac{1}{r-x_1} = - \frac{1}{r} \frac{\frac{x_1}{r}}{1-\frac{x_1}{r}}=-\frac{1}{r} \sum_{k=1}^5\left(\frac{x_1}{r}\right)^k -\frac{1}{r} \frac{\left(\frac{x_1}{r}\right)^6}{1-\frac{x_1}{r}}.
\ee
which holds for all $x_1<r$. 
Therefore, using that $|\psi|^2$ is spherically symmetric so multiplication with an odd function and integration over $\R^3$ gives $0$,  we arrive at 
\be 
\la\psi|\frac{U}{2}\psi\ra =-\frac{1}{4r}\la \psi| \left(\frac{x_1^2}{r^2} + \frac{x_1^4}{r^4}  \right) \psi\ra - \frac{1}{4r}\la \psi| \left(\frac{x_1^6}{r^6}  \frac{1}{1 + \frac{x_1}{r}}  \right) \psi\ra. 
\ee
From \eqref{abgeschnitten} it follows that $|x_1/r|<\frac{1}{4}$ on $\text{supp} \psi$. Thus, with the help of the exponential decay of $\zeta$, we find that there exists $D>0$ such that
\be \label{psiIpsi1}
-\frac{1}{4r}\la \psi| \left(\frac{x_1^2}{r^2} + \frac{x_1^4}{r^4}  \right) \psi\ra -\frac{D}{r^7} \leq \la\psi|\frac{U}{2}\psi\ra  \leq -\frac{1}{4r}\la \psi| \left(\frac{x_1^2}{r^2} + \frac{x_1^4}{r^4}  \right) \psi\ra. 
\ee
and 
\be \label{psiIpsi2}
 -\frac{1}{4r}\la\psi|\left(\frac{x_1^2}{r^2} + \frac{x_1^4}{r^4}  \right) \psi\ra = -\frac{1}{4r}\la\zeta|\left(\frac{x_1^2}{r^2} + \frac{x_1^4}{r^4}  \right) \zeta \ra + O(e^{-cr})=-\frac{1}{r^3}-\frac{18}{r^5} + O(e^{-cr}), 
\ee
where the last step follows from explicitly calculating the integral,  using spherical coordinates and setting $x_1=R \cos(\theta)$. Using \eqref{psiIpsi1} and \eqref{psiIpsi2} we arrive at \eqref{psiIpsiest}.
With the estimates \eqref{Wrest} and \eqref{psiIpsiest} we can conclude  the proof of Theorem \ref{satz}. 
\end{proof}

\begin{Bemerkung}
	We point out possibilities to improve Theorem \ref{satz}.
One can push the decomposition \eqref{geometricdecomposition} not only to order $5$ but to  an arbitrary order. 
This yields an  asymptotic expansion of $\beta \la\psi|\frac{I}{2}\psi\ra$ in powers of $\frac{1}{r}$. 
Following \cite{anapolitanossigal} one can also prove that there exists $\sigma>0$, independent of $\beta$ such that
$$\la\psi|PHP^{\bot}(H^{\bot}-E)^{-1}P^{\bot}HP\psi\ra=-\beta^2 \frac{\sigma}{r^6} + O(\frac{1}{r^7}),$$
which with observations of \cite{barbaroux} can be  improved to 
$$\la\psi|PHP^{\bot}(H^{\bot}-E)^{-1}P^{\bot}HP\psi\ra=-\beta^2 \frac{\sigma}{r^6} + O(\frac{1}{r^8}).$$
In \cite{barbaroux} it was observed that in case of two atoms it is possible to do an expansion of $W(r)$ up to an arbitrary negative power of $r$. We may illustrate here how one can do this with the help of the Feshbach map: Due to the fact that for $P=|\psi \ra \la \psi|$ the assumptions of Theorem \ref{TheoremF} are fulfilled for $\lambda=E$ it is known, see e.g. \cite{feschbach}, that the ground state of the system is given up to normalization by
$$P \psi-(H^\bot-E)^{-1} P^\bot H \psi.$$
Thus replacing the full resolvent with the free resolvent one obtains a test function
 $$P \psi-(H_h^\bot-E_h)^{-1} P^\bot \halb I \psi,$$
which can be given in terms of the free system and is better than $\psi$. Such a test function was used in \cite{paperjannis} to prove upper bounds with error estimates better than the error estimates in \cite{anapolitanossigal}, and later in \cite{anapolitanoslewin} to improve the upper bound on the van der Waals asymptotic of  molecules of Lieb and Thirring in \cite{liebthirring}. One can apply the Feshbach map now with this new test function and by iterating this procedure one obtains  better test functions. Repeating this inductively, yields, in principle, an expansion of $W(r)$ up to an arbitrary negative power of $r$.
All these observations can improve Theorem \ref{satz} but for simplicity of the paper we shall not work them out explicitly. 
\end{Bemerkung}

%------------------------------------------------- HVZ -----------------------------------------------------------

   \section{Proof of Theorem \ref{satzallg}}\label{generalizationmolecule}

   %  In Appendix \ref{prooflimitinfenergy} we prove 
   %  \begin{Lemma}\label{limitinfenergy}
   %  	The limit $\lim_{s \to \infty} E(s,v)$ exists and is independent of $v$.
   %  \end{Lemma}
   
   \begin{proof}
   	In most of the Section we will prove part a) of Theorem 
   	\ref{satzallg}. At the end we will explain how to modify the proof in order to obtain part b). 
    Since many main ideas are very similar to those in the proof of Theorem \ref{satz}, we will sketch the proof and mostly focus on the explanation of the modifications.
   	 We shall use the Feshbach map with $P$ the orthogonal projection onto
   	the cutoff ground state eigenspace of $H_N$ having range
   	\begin{equation}\label{RanP}
 \text{Ran} P = \left\{ h_r^{\otimes N} \psi : \psi \in B\right\},
   	\end{equation}
   	where $h_r$ is the same as in \eqref{abgeschnitten} and $B$ was defined in \eqref{def:B}.
   	We will show, using Condition \eqref{molvermutung} that
   	\begin{equation}\label{IMSboundmol}
   	Q_N H^\bot Q_N-E \geq c >0,
   	\end{equation}
   	where $E=E(r,v)$ is the ground state energy of $Q_N H Q_N|_{\text{Ran}(Q_N)}$, where recall that $Q_N$ was defined in \eqref{def:Qm}. 
   	 Arguing similarly as in the proof of \eqref{roughupperbound}, we find that there exists $c>0$ such that 
   	 \begin{equation}\label{Eupper}
   	 E \leq E_N + \mathcal{O}\left(\frac{1}{r}\right).
   	 \end{equation}
   	 The reason that, unlike \eqref{roughupperbound}, we do not have an exponentially small error, is that we do not have an analogue of Lemma \ref{lemmatrapezoid}
   	 making sure that $H \leq H_N$. 
   	 Using \eqref{ENdef} and \eqref{HNdef} it follows that
   	 $E_N<0$, because in \eqref{HNdef} we have omitted the repulsion terms between the nuclei, which as pointed at the beginning does not change the interaction energy.  This together with \eqref{Eupper} implies that
   	\begin{equation}\label{Eneg}
   	E<0.
   	\end{equation}
when $r$ is not too small and thus 
   	 \begin{equation}
   		E= \inf \sigma (Q_N H Q_N),
   	 \end{equation}
   	namely the restriction onto the range of $Q_N$ in the right hand side can be removed. This helps us to apply the IMS localization formula. 
   	The proof of \eqref{IMSboundmol} is similar to the proof of \eqref{conditionF} and we shall sketch it. 
   	We use the IMS localization formula with the partition of unity $(J_a)_{a \in \{1,2\}^N}$, where for $a=(a_1,...,a_N)  \in \{1,2\}^N$
   	\begin{equation}
   	J_a=J_{a_1} \otimes J_{a_2} \otimes... \otimes J_{a_N},
   	\end{equation}
   	with $J_1, J_2$ defined in \eqref{J1def} and \eqref{J2def}. From \eqref{partition} we see that 
   	\begin{equation}\label{partitiongen}
  \sum_{a \in\{1,2\}^N} J_a^2=1.
   	\end{equation}
   	 Hence, we can apply the
    IMS localization formula to find that
   	$ H= \sum_{a \in \{1,2\}^N} (J_a H J_a - |\nabla J_a|^2)$
   	so that
   	\begin{equation}\label{IMSmol}
   	H= \sum_{a \in \{1,2\}^N} J_a H J_a -\mathcal{O}\left(\frac{1}{r^2}\right).
   	\end{equation}
    	We will now prove that for all $a \in \{1,2\}^N \setminus (2,\dots,2)$ there exists $\delta_a>0$ with
   	\begin{equation}\label{molgap}
   	Q_N J_a H J_a Q_N  > \left(E_N +\delta_a+\mathcal{O}\left(\frac{1}{r}\right)\right)Q_N J_a^2 Q_N.
   	\end{equation} 
Indeed,  if $a \in \{1,2\}^N \setminus (2,\dots,2)$ then $a$ has $k$ times $1$ and $N-k$ times 2, where $k >0$. Thus, we may assume without loss of generality that
   	$J_a=J_1^{\otimes k} \otimes J_2^{\otimes N-k}$. Since the repulsive terms between the first $k$ electrons and the rest $N-k$ are positive we find
   	$$	J_a H J_a \geq J_a \left( A_k \otimes I^{N-k} + I^k \otimes H_{N-k} + \mathcal{O}\left(\frac{1}{r}\right) \right) J_a,$$ where $I^m$ denotes the identity on $m$ particle coordinates and $A_k$ was defined in \eqref{Ak}. 
   	Therefore, using that $Q_N= Q_N (Q_k \otimes Q_{N-k})=  (Q_k \otimes Q_{N-k}) Q_N$ and that $Q_k \otimes Q_{N-k}$ commutes with $J_a$,  we find 
   	\begin{align} \label{lowernonbinding}
   	 Q_N J_a H J_a Q_N  \geq \left(\inf \sigma(Q_k A_k) + \inf \sigma (Q_{N-k} H_{N-k}) +\mathcal{O}\left(\frac{1}{r}\right)\right) Q_N J_a^2 Q_N,
   	\end{align}
   	which together with Condition \eqref{molvermutung} and Lemma \ref{lem:k electrons plate} implies \eqref{molgap}. On the other hand, for $a_0=(2,\dots,2)$ we have  
   	\begin{equation}\label{lowernonbind}
   Q_N J_{a_0} H J_{a_0} Q_N \geq Q_N J_{a_0}  H_N J_{a_0} Q_N +\mathcal{O}\left(\frac{1}{r}\right) J_{a_0}^2\, .
   	\end{equation}
   	Since $J_{a_0}=1$ on the support of $\Psi$, for all $\Psi$ in the range of $P$, one has $P^\bot J_{a_0}= J_{a_0} P^\bot$, and $Q_N$ commutes with $J_{a_0}$ and $P$. 
   	Thus
   	 	\begin{align}\nonumber
   	Q_N P^\bot	J_{a_0} H J_{a_0} & P^\bot Q_N  \geq Q_N J_{a_0} P^\bot Q_N \left(H_N +  \mathcal{O}\left(\frac{1}{r}\right)\right) P^\bot  J_{a_0} Q_N \\ \label{molgapa0}  \geq & \left(E_N+\delta_{a_0}  + \mathcal{O}\left(\frac{1}{r}\right)\right) Q_N P^\bot J_{a_0}^2 P^\bot Q_N,
   	 	\end{align}
   	 	for a $\delta_{a_0}>0$.
Using \eqref{IMSmol}, \eqref{molgap}, \eqref{molgapa0} and  \eqref{partitiongen}, we arrive at 
\begin{align}\nonumber 
Q_N P^\bot	 H P^\bot Q_N  \geq Q_N & P^\bot \left(E_N + \delta + \mathcal{O}\left(\frac{1}{r}\right)\right)P^\bot Q_N \\ \label{IMSprebound} & \geq E_N + \delta + \mathcal{O}\left(\frac{1}{r}\right),
\end{align}
where $\delta = \min_{a \in\{1,2\}^N} \delta_a>0$. Note that the last inequality simply follows from the fact that $E_N + \delta$ is negative -- one can see this by  arguing as in the proof of \eqref{Eneg}. 
 Using \eqref{Eupper} and \eqref{IMSprebound} we arrive at \eqref{IMSboundmol}.
   	
   	In the rest of the proof we identify $H$ with $HQ_N$.
   That
  \begin{equation}\label{existencegroundstate}
  E \in \sigma_{\text{disc}}(H)
  \end{equation} 	
    can be proven similarly as in the case of the hydrogen atom with an HVZ type argument. We will use, however, a faster argument relying on \eqref{IMSboundmol}. If $E$ were in the essential spectrum of $H$, then there would exist a Weyl sequence $\psi_n$ for $H$ and $E$. In particular, we would have 
   	 	\begin{equation}\label{weylseq}
   	 	\la \psi_n, H \psi_n \ra \to E.
   	 	\end{equation}
Since $P$ is a finite rank orthogonal projection, hence compact,  and $\psi_n \rightarrow 0$ weakly, this implies $P \psi_n \rightarrow 0$ and since the operator $HP$ is also bounded we would also have  $H P \psi_n \rightarrow 0$ strongly. But this together with 
\eqref{weylseq} would give $\|P^\bot \psi_n\| \rightarrow 1$
and 
	\begin{equation}
	\la \psi_n, H^\bot \psi_n \ra \to E,
	\end{equation}
contradicting \eqref{IMSboundmol}. 
   
   	Thus \eqref{existencegroundstate} holds, namely $E$ is in the discrete spectrum of $H$ and, in particular, it is an eigenvalue of $H$. From the last observation, \eqref{IMSboundmol} and \eqref{Feshbacheigen} it follows that E is an eigenvalue of $F_P(E)$ and thus there exists $\Psi \in \text{Ran} P$ with
   	\begin{equation}\label{Efesh}
   	E=\langle \Psi, F_P(E) \Psi \rangle=\langle \Psi, H \Psi \rangle - \langle P^\bot H \Psi, (H^\bot-E)^{-1}  P^\bot H \Psi \rangle.
   	\end{equation} 
   	From \eqref{eq:expdecay} and \eqref{RanP} it follows that  $(H_N-E_N) \Psi =\mathcal{O}(e^{-cr})$, which together with  \eqref{Hdec}, \eqref{interen}, and \eqref{Efesh} gives 
   	\begin{equation}\label{intenexpr}
   	W(r,v)=\beta \langle \Psi, \frac{I}{2} \Psi \rangle - \langle P^\bot\frac{I}{2}   \Psi, (H^\bot-E)^{-1}  P^\bot \frac{I}{2} \Psi \rangle + O(e^{-cr}).
   	\end{equation} 
   	Using for $v \in S^2$ the Taylor expansion
   	\begin{equation}
   	\frac{1}{|2rv-z|} =\frac{1}{2r}+\frac{z \cdot v}{4r^2}+ \frac{3(z \cdot v)^2-|z|^2}{16 r^3} + \mathcal{O}\left(\frac{|z|^3}{r^4}\right) \text{ } \forall z \leq \frac{5r}{3},
   	\end{equation}
   	on the support of $\Psi$, the assumptions \eqref{neutral}, \eqref{centermass0} and the definition  \eqref{eqI1} an elementary but lengthy computation gives 
   	\begin{align}\nonumber
   	I_1=\frac{N^2}{r} +  N \sum_{i=1}^{N} \frac{  x_i \cdot v}{2 r^2} +   N \sum_{i=1}^N \frac{3 (x_i \cdot v)^2 - |x_i|^2}{8r^3} \\\label{I1exp} +  N  \sum_{l=1}^M Z_l \frac{3 (y_l \cdot v)^2 - |y_l|^2}{8r^3} + \mathcal{O}\left(\frac{\sum_{i=1}^N |x_i|^3}{r^4}\right).
   	\end{align}
 Using   \eqref{eqI2} we can argue similarly, after observing the equality $x_j^* \cdot v=- x_j \cdot v$, to find 
   	\begin{align}\nonumber 
           I_2=\frac{N^2}{2r} +  N \sum_{i=1}^{N} \frac{  x_i \cdot v}{2r^2}  + N  \sum_{i=1}^N \frac{3 (x_i \cdot v)^2 - |x_i|^2}{8r^3} & \\ \label{I2exp}  +   \frac{3 (\sum_{i=1}^N     x_i \cdot v)^2 + (\sum_{i=1}^N  x_i)\cdot(\sum_{j=1}^N  x_j^*)}{8r^3}  + \mathcal{O}\left(\frac{\sum_{i=1}^N |x_i|^3}{r^4}\right)&,
   	\end{align}
   	and in a similarly way, with the help of \eqref{eqI3},
   	\begin{equation}\label{I3exp}
   	I_3=\frac{N^2}{2r} + N \sum_{l=1}^M Z_l \frac{3 (y_l \cdot v)^2 - |y_l|^2}{8r^3} + \mathcal{O}\left(\frac{\sum_{i=1}^N |x_i|^3}{r^4}\right).
   	\end{equation}
   	Using \eqref{Idec} together with \eqref{I1exp}, \eqref{I2exp} and \eqref{I3exp} yields 
   	\begin{equation*}
   	I=-\frac{3 (\sum_{i=1}^N     x_i \cdot v)^2 + (\sum_{i=1}^N  x_i)\cdot(\sum_{j=1}^N  x_j^*)}{8r^3} + \mathcal{O}\left(\frac{\sum_{i=1}^N |x_i|^3}{r^4}\right),
   	\end{equation*}
   	in the support of $\Psi$. 
   	Let $w=	\sum_{j=1}^N  x_j$. We extend $v$ to an orthonormal basis $v, v_1, v_2$  of $\R^3$. Using $w^* \cdot v=-w \cdot v$, and $w^* \cdot v_i=w \cdot v_i$ for $i=1,2$, as well as 
   	$$ w\cdot w^*= (w\cdot v) (v\cdot w^*)+ (w\cdot v_1) (v_1 \cdot w^*) + (w\cdot v_2) (v_2 \cdot w^*)\, ,$$
   	we arrive at
   	\begin{equation}\label{Iexp}
   	\halb I=-\beta\frac{ (\sum_{i=1}^N     x_i \cdot v)^2 + \left|\sum_{i=1}^N  x_i \right|^2}{16r^3} + \mathcal{O}\left(\frac{\sum_{i=1}^N |x_i|^3}{r^4}\right),
   	\end{equation}
   	in the support of $\Psi$. 
   	We now use the following well--known Lemma, which we prove for convenience of the reader. 
   	\begin{Lemma}\label{feshbachlowest}
   		$\Psi$ is a ground state of $F_P(E)$. In other words, E is the lowest eigenvalue of $F_P(E)$. 
   	\end{Lemma}
   	\begin{proof}	
   		If this were not the case, 
   		there would exist $\tilde{E}<E$, which is the smallest eigenvalue of $F_P(E)$. Note that
   		$F_P(\lambda)$ is a decreasing continuous function of $\lambda$ in $(-\infty,E]$, because for $\lambda_1<\lambda_2 \in (-\infty,E]$ 
   		$$F_P(\lambda_1)-F_P(\lambda_2)=P HP^\bot (H^\bot-\lambda_1)^{-1} (\lambda_1-\lambda_2)(H^\bot-\lambda_2)^{-1} P^\bot H P\le 0,$$
   		 where we also used \eqref{IMSboundmol} to see that $H^\bot-\lambda$ is invertible on the range of $P^\bot$ when $\lambda\le E$. 
   		 Therefore, $g(\lambda):=\inf\sigma(F_P(\lambda))$, the lowest eigenvalue of $F_P(\lambda)$, is also a decreasing continuous function of $\lambda$ in $(-\infty,E]$ with $\tilde{E}=g(E)<E$. But then the intermediate value theorem shows the existence of  $E_0 \in (\tilde{E},E)$ such that $E_0$ is eigenvalue of $F_P(E_0)$. Hence, by \eqref{Feshbacheigen}, $E_0$ would also be an eigenvalue of $H$, contradicting that $E$ is the ground state energy of $H$. 	
   	\end{proof}
   	
   	Using Lemma \ref{feshbachlowest}  together with  \eqref{def:Cv}, \eqref{IMSboundmol}, \eqref{intenexpr}, \eqref{Iexp} and \eqref{eq:expdecay},  
   	we arrive at \eqref{eqn:mainthm}, where Lemma \ref{feshbachlowest}  ensures that $C(v)$ is given by  maximizing the right hand side of \eqref{def:Cv}. This concludes the proof of part a) of Theorem \ref{satzallg}.

   	For part b) equation \eqref{eqn:cor1} follows from the simple observation that if we extend the Taylor expansion
   	\eqref{I1exp}-\eqref{I3exp} to the power of $r^{-4}$ all the terms are odd functions of $x_i,y_j$. Due to the invariance of the molecule and the one electron density with respect to the map $x \to -x$ we obtain an integral of an odd function which has to be zero.
   	
   	To obtain \eqref{eqn:cor2} we have to modify the argument similarly as in proving \eqref{psiIpsiest} for the hydrogen atom, after observing that \eqref{Idec}-\eqref{eqI3} simplifies 
   	because we have only one nucleus at 0. The only additional observation that we need is that, Newton's Theorem implies that for any ground state $\psi$ of the atom we have
   	\begin{equation}
\int\frac{1}{|-x_i + 2r v + x_j^*|}|\psi(x_1,\dots,x_N)|^2 dx_1 \dots dx_N=\frac{1}{2r}, \quad \forall i \neq j. 
    \end{equation}
   \end{proof}	
   
\section{Sketch of proof of the necessity of the binding Condition \eqref{molvermutung}}\label{proofnecessity}

Here we sketch how to prove that \eqref{molvermutung} is a necessary
Condition for Theorem \ref{satzallg}. First of all we remark that the general definition of the interaction energy is not given by \eqref{interen} but by
\begin{equation}\label{interengen}
W(r,v)= E(r,v) - E(\infty),
\end{equation}
where 
\begin{equation}\label{Einfty}
 E(\infty) = \lim_{r \rightarrow \infty} E(r,v),
\end{equation}
so the interaction energy measures, if it is negative, how much energy it costs to separate the system. That under \eqref{molvermutung} the equality $E(\infty)=E_N$ holds, follows immediately from our results. Suppose now  that \eqref{molvermutung} does not hold. We will prove that
\begin{equation}\label{Einftyest}
E(\infty) =  \min_{k \in \{1,\dots, N\}} \Big( \inf \sigma(Q_{N-k} H_{N-k}) 
+ k \beta^2 E_{e^-}\Big),
\end{equation}
and that $\exists C_1,C_2>0$ such that
\begin{equation}\label{interenest}
-\frac{C_2}{r} \leq W(r,v) \leq -\frac{C_1}{r}.
\end{equation}
We will prove \eqref{Einftyest} and \eqref{interenest} at once by proving that $\exists C_1,C_2>0$ such that
\begin{equation}\label{interenres}
-\frac{C_2}{r} \leq E(r,v) - \min_{k \in \{1,\dots, N\}} \Big( \inf \sigma(Q_{N-k} H_{N-k}) 
+ k \beta^2 E_{e^-}\Big) \leq -\frac{C_1}{r}.
\end{equation}
This can be done using ideas of the proof of Theorem 1.1 in \cite{morgansimon}. Note that we can not apply  the Feshbach map here as it is not clear if the system has a ground state. To prove the left inequality of \eqref{interenres} we observe that \eqref{lowernonbinding} and Lemma \ref{lem:k electrons plate} imply that 
	\begin{align}\nonumber
	Q_N J_a H J_a Q_N & \\\label{lowernonbinding3} \geq \left( \min_{k \in \{1,\dots, N\}} \Big( \inf \sigma(Q_{N-k} H_{N-k}) + k \beta^2 E_{e^-}\Big) +\mathcal{O}\left(\frac{1}{r}\right)\right) & Q_N J_a^2 Q_N, \text{ } \forall a \in \{1,2\}^N \setminus (2,\dots,2).
	\end{align}

Moreover from \eqref{lowernonbind} and the negation of \eqref{molvermutung} we find that

\begin{equation}\label{lowernonbinding4}
Q_N J_{a_0} H J_{a_0} Q_N  \geq \left( \min_{k \in \{1,\dots, N\}} \Big( \inf \sigma(Q_{N-k} H_{N-k}) + k \beta^2 E_{e^-}\Big) +\mathcal{O}\left(\frac{1}{r}\right)\right) Q_N J_{a_0}^2 Q_N.
\end{equation}
Inserting \eqref{lowernonbinding3} and \eqref{lowernonbinding4} in \eqref{IMSmol} and using \eqref{partitiongen} we arrive at the left inequality of  \eqref{interenres}. To prove the right inequality of \eqref{interenres} one just has to choose a sequence of appropriate test functions as follows: Pick $k_0$ such that
	\begin{equation}
 \inf \sigma(Q_{N-k_0} H_{N-k_0}) + k_0 \beta^2 E_{e^-} = \min_{k \in \{1,\dots, N\}} \Big( \inf \sigma(Q_{N-k} H_{N-k}) + k \beta^2 E_{e^-}\Big).
	\end{equation} 
Now choose a cut off ground state of the positive ion with Hamiltonian $Q_{N-k_0} H_{N-k_0}$ and place the other electrons far away from each other and from the ion. Then the Coulomb attraction of the positive ion with its mirror image will dominate. We omit the details. 
%----------------------------------------------------------------------------------------------------------------------------
%-------------------------------------------------------APPENDIX---------------------------------------------------------
%----------------------------------------------------------------------------------------------------------------------------
\appendix

%----------------------------------------------------- Potential Ladung q ------------------------------------------------------------------------

\section{Potential created by a charge q outside a dielectric half--space}
\label{spiegelladung}
In this section we  illustrate how to derive the Hamiltonian $H(r,v)$ in a case of a dielectric or perfectly  conducting  plate. We consider two infinite dielectric media with permittivities 
$\epsilon_1$ and $\epsilon_2$, that have an infinite plane as their interface. We use coordinates such that the first component vanishes at the interface. 

\noindent 
First we derive the Green's function. For a more detailed derivation we refer to \cite{schwinger} Chapters 12-14.
Then with the help of the Green's  function we derive the interaction energy for a charge distribution with several charges. Thus the derivation works for the general case of a molecule interacting with a dielectric plate. 

\begin{center}
	\includegraphics[scale=0.3]{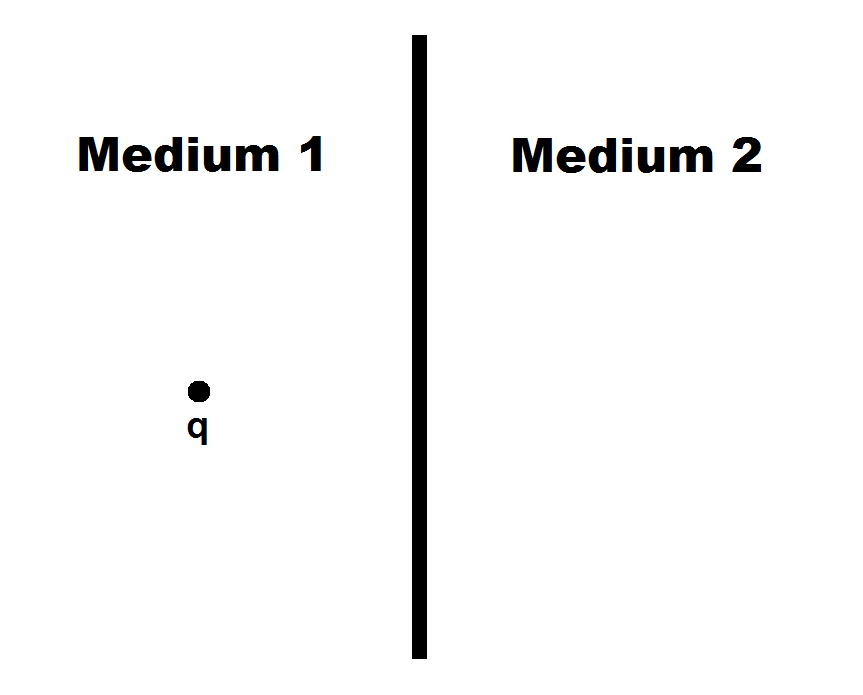}\\
	\small Figure 5
\end{center}

If a charge $q$ with $q=1$ is in Medium 1 at the position $\xp$, see Figure 5, the Green's function must satisfy the conditions 
\be \begin{cases}
	\text{  }\text{  }-\epsilon_1\Delta G_1(\x,\xp)=\delta(\x-\xp), & \text{for } \x \text{ in Medium } 1\\
	-\epsilon_2\Delta G_2(\x,\xp)=0,& \text{for } \x \text{ in Medium } 2.
\end{cases}\ee
We use the method of mirror images and make the ansatz 

\be
G_1(\x,\xp)=\frac{1}{4\pi  \epsilon_1}\bigg(\frac{1}{\mid \x-\xp\mid}+\frac{A}{\mid \x-\xp_s\mid}\bigg),
\ee

\be
G_2(\x,\xp)=\frac{1}{4\pi  \epsilon_2}\frac{B}{\mid \x-\xp\mid},
\ee
where $\xp_s$ is the position of the mirror image of the charge, and $A$ and $B$ are to be determined.\\

Let $v=(0, v_2, v_3 )^t$ now be a point on the interface. Then 
\be
G_1(v,\xp)=G_2(v,\xp)\, ,
\ee
thus also 
\be
\frac{1}{\epsilon_1}\bigg(\frac{1}{\mid v-\xp\mid}+\frac{A}{\mid v-\xp_s\mid}\bigg)=\frac{1}{\epsilon_2}\frac{B}{\mid v-\xp\mid}
\ee

\be \label{AB}
\Leftrightarrow\frac{1}{\epsilon_1}(1+A)=\frac{B}{\epsilon_2},
\ee
because $v_1=0$, which implies $\frac{1}{\mid v-\xp\mid}=\frac{1}{\mid v-\xp_s\mid}$.\\

Due to the fact that the normal component of the electric displacement has to be continuous the Green's function must also satisfy the boundary conditions
\be
\epsilon_1 \frac{\partial G_1(x,y)}{\partial x_1}\bigg|_{x=v}=\epsilon_2\frac{\partial G_2(x,y)}{\partial x_1}\bigg|_{x=v}
\ee
%\be
%\Leftrightarrow -\frac{\epsilon_1}{4\pi  %\epsilon_1}\frac{A\xp_1}{\mid v-\xp\mid^3}=\frac{\epsilon_2}{4\pi %\epsilon_2}\bigg(\frac{-\xp_1}{\mid %v-\xp_1\mid^3}+\frac{B\xp_1}{\mid v-\xp_s\mid^3}\bigg)
%\ee
%\be
%\Leftrightarrow -\frac{A\xp_1}{\mid %v-\xp\mid^3}=\frac{-\xp_1}{\mid v-\xp\mid^3}+\frac{B\xp_1}{\mid %v-\xp_s\mid^3}
%\ee
\be
\Leftrightarrow A=1-B.
\ee
This together with \eqref{AB} gives

\be
B=\frac{2\epsilon_2}{\epsilon_1+\epsilon_2} ,\text{  } A= \frac{\epsilon_1-\epsilon_2}{\epsilon_1+\epsilon_2}\, ,
\ee
hence 
\be \label{G1}
G_1(\x,\xp)=\frac{1}{4\pi \epsilon_1}\bigg(\frac{1}{\mid \x-\xp\mid}+\frac{\epsilon_1-\epsilon_2}{\epsilon_1+\epsilon_2}\frac{1}{\mid \x-\xp_s\mid}\bigg)
\ee
\be
G_2(\x,\xp)=\frac{2}{4\pi (\epsilon_1+\epsilon_2)}\frac{1}{\mid \x-\xp\mid}.
\ee

%The second condition is:
%\be  D_{n_1}(r)=D_{n_2}(r) \ee
%\be \Leftrightarrow %\epsilon_1E_{n_1}(r)=\epsilon_2E_{n_2}(r) \ee
%\be \label{Uebergangsbed} \Leftrightarrow %\epsilon_1\frac{\partial\Phi_{1}(r)}{\partia %x}=\epsilon_2\frac{\partial\Phi_{2}(r)}{\part%ial x} \ee
%where the normal vectors at the interface %point in the x direction. So the potential %is continuous at the boundary.\\
\smallskip

\noindent 
Now we are going to derive the interaction potential with the help of the Green's function. We do this in a more general setting of an interacting system, following \cite{schwinger} Chapter 15. We assume that
\smallskip 

1) We know the full Green's function $G^w(\x,\xp)$ at least in the right half--space.

2) $G^w(\x,\xp)=G_0(\x,\xp) + G_d(\x,\xp)$
where $G_0$ is the free Green's function and $G_d(\x,\xp)$ the perturbation and that 
$\lim_{x \to \xp} G_d(\x,\xp)$ exists.

In our case 

\begin{equation}\label{Gwg0}
G^w(\x,\xp)=G_1(\x,\xp), \quad G_0(\x,\xp)= \frac{1}{4\pi \epsilon_1}\frac{1}{\mid \x-\xp\mid},
\end{equation}
so $G_0$ would be the Green's function if there were no Medium 2. 
Thus by \eqref{G1} and \eqref{Gwg0} we have 
\begin{equation}
G_d(\x,\xp):=G_1(\x,\xp) - G_0(\x,\xp)=
\frac{1}{4\pi \epsilon_1}\frac{\epsilon_1-\epsilon_2}{\epsilon_1+\epsilon_2}\frac{1}{\mid \x-\xp_s\mid},
\end{equation}
and, in particular,
 \begin{equation}
 \lim_{x \to \xp} G_d(\x,\xp)=
 \frac{1}{4\pi \epsilon_1}\frac{\epsilon_1-\epsilon_2}{\epsilon_1+\epsilon_2}\frac{1}{\mid \xp-\xp_s\mid}.
 \end{equation}
 Thus in our case the Assumptions 1 and 2 are satisfied. 
 \smallskip
 
 \noindent 
 Now we will explain how the Hamiltonian is derived and in particular how the presence of the factor $\frac{1}{2}$ in the interaction of charges--mirror charges occurs. 
  In terms of the Green's function, the electrostatic energy 
  of the system of a charge distribution $\rho$ is given by 
 \begin{equation}
 E^w(\rho)= D^w(\rho,\rho)= \frac{1}{2} \int_V \int_V dx_1 dx_2 G^w(x_1,x_2) \rho(x_1) \rho(x_2),
 \end{equation}
where $V$ is the region occupied by Medium 1. In the absence of Medium 2, 
 the free electrostatic energy of a charge distribution $\rho$ is given by 
 \begin{equation}
 E_0(\rho)= D_0(\rho,\rho)= \frac{1}{2} \int_V \int_V dx_1 dx_2 G_0(x_1,x_2) \rho(x_1) \rho(x_2).
 \end{equation}
We consider now a charge distribution
\begin{equation}
\rho=\sum_{j=1}^N \rho_j.
\end{equation}
The \emph{interaction energy} of the entire system is then given by the energy difference 
\begin{align}\nonumber
E_{\text{int}}(\rho):= & E^w(\rho) - \sum_{j=1}^N E_0(\rho_j) = D^w\bigg(\sum_{j=1}^N \rho_j, \sum_{j=1}^N \rho_j\bigg) - \sum_{j=1}^N D_0(\rho_j,\rho_j) \\  \label{intendec}
 & =\sum_{i \neq j}D^w(\rho_i, \rho_j) + \sum_{j=1}^N \big(D^w(\rho_j, \rho_j) -D_0(\rho_j, \rho_j)\big).
\end{align}
In the limit $\rho_j \rightarrow q_j \delta_{x_j}$ we have
\begin{align}\nonumber
 D^w(\rho_i, \rho_j) & \to \frac{1}{2} G^w(x_i,x_j)\\\label{dwijlim} 
 	&= \frac{1}{2} \underbrace{q_i q_j G_0(x_i,x_j)}_{(1)}   
 		+ \underbrace{\frac{1}{2} q_i q_j G_d(x_i,x_j)}_{(2)}.
\end{align}
where (1) is the direct Coulomb  interaction of the the charges 
$ q_i,q_j$ and (2) is the interaction of $q_i$ with the mirror image of the charge $q_j$. 

We will see later that, when summed over all $i \neq j$ the term 
$\frac{1}{2} q_i q_j G_0(x_i,x_j)$ appears twice, so we get rid of the 
factor $\frac{1}{2}$. Moreover, while the limit of $D^w(\rho_i, \rho_j)$ does not exist when $\rho_j$ approaches a point charge, the difference of the electrostatic energies $D^w(\rho_j, \rho_j) -D_0(\rho_j, \rho_j)$ converges since 
$G^w=G_0 + G_d$ in the right half--space. In particular,  
$$
 D^w(\rho_j, \rho_j) - D_0(\rho_j,\rho_j)= \frac{1}{2} \int_V \int_V dx_1 dx_2 G_d(x,y) \rho_j(x) \rho_j(y).$$
 Hence
 \begin{equation}\label{Gdjjlim}
 D^w(\rho_j, \rho_j) - D_0(\rho_j,\rho_j) \to \underbrace{\frac{1}{2} q_j^2 G_d(x_j,x_j)}_{\text{interaction of charge } q_j \text{ with its own mirror charge}}.
\end{equation}
Using \eqref{intendec}, \eqref{dwijlim} and \eqref{Gdjjlim}, we find 
\begin{equation}
E_{\text{int}}=\underbrace{
	\sum_{1 \leq i < j \leq N} q_i q_j G_0(x_i,x_j)}_{\text{direct Coulomb interaction terms}} + \underbrace{
	\frac{1}{2}\left( \sum_{i  \neq j} q_i q_j G_d(x_i,x_j) +   \sum_{j=1}^N q_j^2 G_d(x_j,x_j) \right)}_{ \text{interaction of charges with mirror charges} }.
\end{equation}
This means that the potential of the system dielectric-plate/point-charges can be easily computed with the help of the mirror charges. 
\smallskip

It is also interesting to note that in the limit $\epsilon_2 \rightarrow \infty$ which is relevant  for perfect conductors, we have $\frac{\epsilon_1-\epsilon_2}{\epsilon_1+\epsilon_2}\rightarrow -1$. This is exactly the potential obtained with the classical method of mirror images.

\section{A trace Theorem for functions in $H^1(\R^n)$ restricted on a hyperplane}\label{App:trace}
Such a trace Theorem is certainly well--known	 in the literature. Since most books deal with the harder case 
of domains with suitable smooth boundaries, we provide here the proof of the considerably simpler case of a half--space, for convenience of the reader. 
A point in $\R^n$ is denoted by $(x,x')$, where $x \in \R$ and $x' \in \R^{n-1}$. With $\mathcal{S}(\R^m)$ we denote the set of Schwartz functions in $\R^m$.

\begin{Theorem}\label{tracethm}
There exists a unique linear continuous map 
$T: H^1(\R^n) \rightarrow L^2(\R^{n-1})$
with $T f(x')= f(0,x')$ for all $f \in \mathcal{S}(\R^n)$ and 
\begin{equation}\label{traceineq}
\|T f\|_{L^2} \leq \sqrt{\pi} \|f\|_{H^1}.
\end{equation}
\end{Theorem}
\begin{proof}
 Let $f \in \mathcal{S}(\R^n)$. Then by the Fourier inversion formula we have 
 $$f(0,x')= \int_\R (\mathcal{F}_1 f)(\xi, x') d\xi,$$
where $(\mathcal{F}_1 f)$ denotes the Fourier transformation of $f$ only with respect to the first variable. Thus with the help of the triangle inequality for the $L^2$ norm in the $x'$ integral 
 $$\Big(\int_{\R^{n-1}} |f(0,x')|^2 dx'\Big)^{\frac{1}{2}}\leq \int_{\R} \Big(\int_{\R^{n-1}} |(\mathcal{F}_1 f)(\xi, x')|^2 d x'\Big)^\frac{1}{2} d \xi,$$
 or
  $$\|Tf\|_{L^2} \leq \int_\R (1+|\xi|^2)^{-\frac{1}{2}} \Big(\int_{\R^{n-1}} (1+|\xi|^2)|(\mathcal{F}_1 f)(\xi, x')|^2 d x'\Big)^\frac{1}{2} d \xi.$$
  Thus applying Cauchy-Schwarz we find 
   $$\|Tf\|_{L^2} \leq \Big(\int_\R (1+|\xi|^2)^{-1} d \xi \Big)^\frac{1}{2}  \Big(\int_{\R \times \R^{n-1}} (1+|\xi|^2)|(\mathcal{F}_1 f)(\xi, x')|^2 d x' d \xi \Big)^\frac{1}{2},$$
   which together with  $\int_\R (1+\xi^2)^{-1} d \xi = \pi$ and Plancherel's Theorem  gives \eqref{traceineq} for all $f \in \mathcal{S}(\R^n)$. Thus the operator $T$ can be uniquely extended on the whole $H^1(\R^n)$ and its extension also satisfies  the bound  \eqref{traceineq}.
\end{proof}

\begin{Theorem}\label{tracethm2}
	We consider $f \in H^1(\R^n)$ which is odd in the $x$ variable. Then $Tf=0$ and $f|_{\R_+\times \R^{n-1}} \in H_0^1(\R_+\times \R^{n-1})$. 
\end{Theorem}
\begin{proof}
 The function $f$ can be approximated by a sequence of Schwartz functions $f_n$ which are odd in the $x$ variable. That it can be assumed that $f_n$ is odd comes from the fact that the odd part of $f_n$ defined by $ \frac{f_n(x,x')-f_n(-x,x')}{2}$ is closer to $f$ in the $H^1$ norm than the function $f_n$ itself.
  Then $f_n(0,x')=0$ and with the help of Theorem \ref{tracethm} it follows that $Tf=0$. 
 
 Observe now that for $g \in \mathcal{S}(\R^n)$ and $x>0$
 $$g(x,x')-g(0,x')= \int_0^x \frac{\partial}{\partial x} g(s,x') ds,$$
 which together with  the triangle inequality for the $L^2$ norm in the $s$ integral gives
 \begin{align}\nonumber
  \Big(\int_{\R^{n-1}} |g(x,x')&-g(0,x')|^2 dx'\Big)^\frac{1}{2} \\\nonumber \leq &\int_0^x \Big( \int_{\R^{n-1}} \big|\frac{\partial}{\partial x}  g(s,x')\big|^2 dx'\Big)^\frac{1}{2} ds.
   \end{align}
  Applying now Cauchy-Schwarz we obtain 
   \begin{align}\nonumber
  \Big(\int_{\R^{n-1}} |g(x,x')&-g(0,x')|^2 dx'\Big)^\frac{1}{2} \\ \nonumber \leq& \sqrt{x}  \bigg(\int_0^x  \int_{\R^{n-1}} \big|\frac{\partial}{\partial x}  g(s,x')\big|^2 dx' ds\bigg)^\frac{1}{2}.
     \end{align}
  Thus if $Tg=0$ then
    \begin{equation*}
   \int_{\R^{n-1}} |g(x,x')|^2 dx'\leq x  \int_0^x  \int_{\R^{n-1}} \big|\frac{\partial}{\partial x}  g(s,x')\big|^2 dx' ds.
      \end{equation*}
    With the help of Theorem \ref{tracethm} and with approximation by Schwartz functions it follows that f satisfies the same inequality, namely 
      \begin{equation*}
     \int_{\R^{n-1}} |f(x,x')|^2 dx' \leq x  \int_0^x  \int_{\R^{n-1}} \big|\frac{\partial}{\partial x}  f(s,x')\big|^2 dx' ds.
        \end{equation*}
     Therefore, since $f \in H^1$,  we find 
      \begin{equation}\label{ineqfxl2}
       \int_{\R^{n-1}} |f(x,x')|^2 dx'\leq x c(x), \text{ with } \lim_{x\rightarrow 0} c(x)=0.
      \end{equation}
      For the rest of the proof we identify $f$ with $f|_{\R_+ \times \R^{n-1}}$.
       We will prove that $f \in H_0^1(\R_+ \times \R^{n-1})$. 
    Let $\chi: \R_+ \to [0,1]$ be a $C^\infty$ nondecreasing function with $\chi(x)=0$ if $x \leq 1$ and $\chi(x)=1$ if $x \geq 2$. Let $\chi_n(x)=\chi(nx)$. Obviously $\chi_n f \in H_0^1(\R_+ \times \R^{n-1})$ for all $n \in \N$ and thus proving that $f \in H_0^1(\R_+ \times \R^{n-1})$ reduces to proving that $\chi_n f \rightarrow f$ in the $H^1$ norm. From the dominated convergence theorem it follows immediately that 
     $\chi_n f \rightarrow f$ in $L^2$ and
      $\chi_n \nabla f \rightarrow \nabla f$ in $L^2$. Thus it suffices to prove that
      $\chi_n' f \rightarrow 0$ in $L^2$. Indeed we have
      \begin{equation*}
      \|\chi_n' f\|_{L^2}^2=\int_{\R_+ \times \R^{n-1}} |f(x,x')|^2 n^2|\chi'(nx)|^2 dx dx'.
      \end{equation*} 
      Consequently the change of variable $y=nx$ together with Fubini's theorem gives
       \begin{equation*}
       \|\chi_n' f\|_{L^2}^2=\int_{R_+} |\chi'(y)|^2 n\Big(\int_{\R^{n-1}} |f\big(\frac{y}{n},x'\big)|^2 dx' \Big) dy.
       \end{equation*} 
       Hence, using  \eqref{ineqfxl2} we arrive at
        \begin{equation}
        \|\chi_n' f\|_{L^2}^2 \leq \int_{\R_+} |\chi'(y)|^2 y c\Big(\frac{y}{n}\Big) dy \rightarrow 0,
        \end{equation}
        because $y \leq 2$ on supp $\chi'$. This concludes the proof of Theorem \ref{tracethm2}.
\end{proof}

\end{document}